%% file: execution_time.tex
\documentclass[acmtocl]{acmtrans2m}

\usepackage{amsmath, amsfonts, amssymb}
\usepackage{stmaryrd}
\usepackage{bussproofs}

\newtheorem{theorem}{Theorem}[section]
\newtheorem{corollaire}[theorem]{Corollary}
\newtheorem{prop}[theorem]{Proposition}
\newtheorem{definition}[theorem]{Definition}
\newtheorem{rem}[theorem]{Remark}
\newtheorem{ex}[theorem]{Example}
\newtheorem{lemme}[theorem]{Lemma}
\newtheorem{fait}[theorem]{Fact}

\title{Execution Time of $\lambda$-Terms via Denotational Semantics and Intersection Types}
\author{DANIEL DE CARVALHO \\ Universit\`a di Roma Tre}
\begin{abstract}
The multiset based relational model of linear logic induces a semantics of the type free $\lambda$-calculus, which corresponds to a non-idempotent intersection type system, System $R$. We prove that, in System $R$, the size of the type derivations and the size of the types are closely related to the execution time of $\lambda$-terms in a particular environment machine, Krivine's machine.

\end{abstract}



\begin{document}

\begin{bottomstuff}
Author's address: D. de Carvalho, Dipartimento di Filosofia, Facolt\`a di Lettere e Filosofia, Universit\`a Roma Tre, Via Ostiense 236, 00146 Roma, Italy
\end{bottomstuff}
\maketitle

\section{Introduction}

This paper presents a work whose aim is to obtain information on execution time of $\lambda$-terms by semantic means. 

By execution time, we mean the number of steps in a computational model. As in \cite{thomas}, the computational model considered in this paper will be Krivine's machine, a more realistic model than $\beta$-reduction. Indeed, Krivine's machine implements (weak) head linear reduction: in one step, we can do at most one substitution. 
In this paper, we consider two variants of this machine : the first one (Definition~\ref{Def Krivine head}) computes the head-normal form of any $\lambda$-term (if it exists) and the second one (Definition~\ref{Def Krivine normal}) computes the normal form of any $\lambda$-term (if it exists).

The fundamental idea of denotational semantics is that types should be interpreted as the objects of a category $ \mathbb{C} $ and terms should be interpreted as arrows in $\mathbb{C}$ in such a way that if a term $t$ reduces to a term $t'$, then they are interpreted by the same arrow. By the Curry-Howard isomorphism, a simply typed $\lambda$-term is a proof in intuitionistic logic and the $\beta$-reduction of a $\lambda$-term corresponds to the cut-elimination of a proof. Now, the intuitionistic fragment of linear logic \cite{LinearLogic} is a refinement of intuitionistic logic. This means that when we have a categorical structure $ (\mathbb{C}, \ldots) $ for interpreting intuitionistic linear logic, we can derive a category $\mathbb{K} $ that is a denotational semantics of intuitionistic logic, and thus a denotational semantics of $\lambda$-calculus.

Linear logic has various denotational semantics; one of these is the multiset based relational model in the category \textbf{Rel} of sets and relations with the comonad associated to the finite multisets functor (see \cite{Lorenzo} for interpretations of proof-nets and Appendix of \cite{phase} for interpretations of derivations of sequent calculus). In this paper, the category $\mathbb{K}$ is a category equivalent to the Kleisli category of this comonad. The semantics we obtain is \emph{non-uniform} in the following sense: the interpretation of a function contains information about its behaviour on chimerical arguments (see Example \ref{example : non-uniformity} for an illustration of this fact). As we want to consider type free $\lambda$-calculus, we will consider $\lambda$-algebras in $ \mathbb{K} $. We will describe semantics of $\lambda$-terms in these $\lambda$-algebras as a logical system, using intersection types. 

The intersection types system that we consider (System $R$, defined in Subsection~\ref{subsection:System R}) is a reformulation of that of \cite{CDV}; in particular, it lacks idempotency, as System $ \lambda $ in \cite{Kfoury} and System $ \mathbb{I} $ in \cite{NeeMairson} and contrary to System $ \mathcal{I} $ of \cite{KMTW}. So, we stress the fact that the semantics of \cite{CDV} can be reconstructed in a natural way from the finite multisets relational model of linear logic using the Kleisli construction. 

If $t'$ is a $\lambda$-term obtained by applying some reduction steps to $t$, then the semantics $\llbracket t' \rrbracket$ of $t'$ is the same as the semantics $\llbracket t \rrbracket$ of $t$, so that from $\llbracket t \rrbracket$, it is clearly impossible to determine the number of reduction steps leading from $t$ to $t'$. Nevertheless, if $v$ and $u$ are two closed normal $\lambda$-terms, we can wonder
\renewcommand{\labelenumi}{\arabic{enumi})}
\begin{enumerate}
\item Is it the case that the $\lambda$-term $(v)u$ is (head) normalizable?
\item If the answer to the previous question is positive, what is the number of steps leading to the (principal head) normal form?
\end{enumerate}
The main point of the paper is to show that it is possible to answer both questions by only referring to the semantics $\llbracket v \rrbracket$ and $\llbracket u \rrbracket$ of $v$ and $u$ respectively. The answer to the first question is given in Section~\ref{section:qualitative} (Corollary \ref{cor:qualitative}) and it is a simple adaptation of well-known results. 
The answer to the second question is given in Section~\ref{section:quantitative}.

The paper \cite{Ronchi} presented a procedure that computes a normal form of any $\lambda$-term (if it exists) by finding its principal typing (if it exists). In Section~\ref{section:quantitative}, we present some quantitative results about the relation between the types and the computation of the (head) normal form. In particular, we prove that the number of steps of execution of a $\lambda$-term in the first machine (the one of Definition~\ref{Def Krivine head}) is the size of the least type derivation of the $\lambda$-term in System $R$ (Theorem \ref{Theorem head}) and we prove a similar result (Theorem \ref{theorem:normal_bound}) for the second machine (the one of Definition~\ref{Def Krivine normal}). We end by proving truly semantic measures of execution time in Subsection~\ref{semantics} and Subsection~\ref{subsection:exact_number}.

Note that even if this paper, a revised version of \cite{execution_old}, concerns the $\lambda$-calculus and Krivine's machine, we emphasize connections with proof nets of linear logic. Due to these connections, we conjectured in \cite{these} that we could obtain some similar results relating on the one hand the length of cut-elimination of nets with some specific strategy 
and on the other hand the size of the results of experiments. This specific strategy should be a strategy that mimics the one of Krivine's machine and that extends a strategy defined in \cite{Marco} for a fragment of linear logic. 
This work has been done in \cite{hal} by adapting our work for the $\lambda$-calculus. But it is still difficult to compare both works, because the syntax of proof nets we considered makes that a cut-elimination step is not as elementary as a reduction step in Krivine's machine.

In conclusion, we believe that this work can be useful for implicit characterizations of complexity classes (in particular, the PTIME class, as in \cite{DLAL}) by providing a semantic setting in which 
quantitative aspects can be studied, while taking some distance with the syntactic details.

In summary, Section~\ref{section : Krivine} presents Krivine's machine, Section~\ref{section : semantics} the semantics we consider and Section~\ref{section : System R} the intersection type system induced by this semantics, namely System $R$; Section~\ref{section:qualitative} gives the answer to question 1) and Section~\ref{section:quantitative} the answer to question 2).

\paragraph*{Notations}
We denote by $ \Lambda $ the set of $\lambda$-terms, by $ \mathcal{V} $ the set of variables and, for any $\lambda$-term $ t $, by $ FV(t) $ the set of free variables in $ t $.

We use Krivine's notation for $ \lambda $-terms: the $\lambda$-term $ v $ applied to the $\lambda$-term $ u $ is denoted by $ (v)u $. We will also denote $(\ldots ((v)u_1)u_2 \ldots)u_k$ by $(v)u_1 \ldots u_k$.

We use the notation $[ \: ]$ for multisets while the notation $\{ \: \}$ is, as usual, for sets. For any set $A$, we denote by $\mathcal{M}_\textrm{fin}(A)$ the set of finite multisets $ a $ whose support, denoted by $ \textsf{Supp}(a) $, is a subset of $ A $. 
For any set $A$, for any $n \in \mathbb{N}$, we denote by $\mathcal{M}_n(A) $ the set of multisets of cardinality $n$ whose support is a subset of $A$. 
The pairwise union of multisets given by term-by-term addition of multiplicities is denoted by a $+$ sign and, following this notation, the generalized union is denoted by a $\sum$ sign. The neutral element for this operation, the empty multiset, is denoted by $[]$. 

For any set $A$, for any $n \in \mathbb{N}$, for any $a \in \mathcal{M}_n(A)$, we set 
$$\mathfrak{S}(a) = \{ (\alpha_1, \ldots, \alpha_n) \in A^n \: / \: a = [\alpha_1, \ldots, \alpha_n] \} \enspace .$$

\input{execution_time_Krivine}

\input{execution_time_semantics}

\input{execution_time_intersection}

\input{execution_time_qualitative}

\input{execution_time_quantitative}

\paragraph*{Acknowledgements.} 
This work is partially the result of discussions with Thomas Ehrhard: I warmly thank him. I also thank Patrick Baillot, Simona Ronchi della Rocca and Kazushige Terui too for stimulating discussions.

\bibliographystyle{acmstrans}

\end{document}

%% file: execution_time_Krivine.tex
\section[Krivine's machine]{Krivine's machine}\label{section : Krivine}

We introduce two variants of a machine presented in \cite{KAM} that implements call-by-name. More precisely, the original machine performs weak head linear reduction, whereas the machine presented in Subsection \ref{subsection:head} performs head linear reduction. Subsection \ref{subsection:beta} slightly modifies the latter machine as to compute the $\beta$-normal form of any normalizable term.

\subsection{Execution of States}\label{subsection:states}

We begin with the definitions of the set $ \mathcal{E} $ of environments and of the set $ \mathcal{C} $ of closures.
 
Set $ \mathcal{E} = \bigcup_{p \in \mathbb{N}} \mathcal{E}_p $ and set $ \mathcal{C} = \bigcup_{p \in \mathbb{N}} \mathcal{C}_p $, where $ \mathcal{E}_p $  and $ \mathcal{C}_p $ are defined by induction on $ p $:
\begin{itemize}
\item If $ p = 0 $, then $ \mathcal{E}_p = \{ \emptyset \} $ and $ \mathcal{C}_p = \Lambda \times \{ \emptyset \} $.
\item $ \mathcal{E}_{p+1} $ is the set of partial maps $\mathcal{V} \rightarrow \mathcal{C}_p$, whose domain is finite, and $ \mathcal{C}_{p+1} = \Lambda \times \mathcal{E}_{p+1}  $.
\end{itemize}

For $ e \in \mathcal{E} $, we denote by $ \textsf{d}(e) $ the least integer $ p $ such that $ e \in \mathcal{E}_p $. 

For $ c = (t, e) \in \mathcal{C} $, we define $\overline{c} = t[e] \in \Lambda$ by induction on $\textsf{d}(e)$:
\begin{itemize}
\item If $ \textsf{d}(e)=0 $, then $ t[e]=t $.
\item Assume $ t[e] $ defined for $ \textsf{d}(e) = d $. If $ \textsf{d}(e) = d+1 $, then  $ t[e] = t[\overline{c_1}/x_1, \ldots,  \overline{c_m}/x_m] $, with $ \{ x_1, \ldots, x_m \} = \textsf{dom}(e) $ and, for $ 1 \leq j \leq m $, $ e(x_j) = c_j $.
\end{itemize}

A \emph{stack} is a finite sequence of closures. If $ c_0 $ is a closure and $ \pi = $ $ (c_1, \ldots, c_q) $ is a stack, then $ c_0 . \pi $ will denote the  stack $ (c_0, \ldots, c_q) $. We will denote by $ \epsilon $ the empty stack. 

A \emph{state} is a non-empty stack. If $ s = (c_0,  \ldots, c_q)$ is a state, then $ \overline{s} $ will denote the $\lambda$-term $ (\overline{c_0}) \overline{c_1} \ldots \overline{c_q} $. 

\begin{definition}\label{definition : variable convention}
We say that a $\lambda$-term $ t $ \emph{respects the variable convention} if any variable is bound at most once in $ t $.

For any closure $ c = (t, e) $, we define, by induction on $ \textsf{d}(e) $, what it means for $ c $ to \emph{respect the variable convention}:
\begin{itemize}
\item if $ \textsf{d}(e)=0 $, then we say that $ c $ respects the variable convention if, and only if, $ t $ respects the variable convention~;
\item if $ c=(t, \{ (x_1, c_1), \ldots, (x_m, c_m) \} ) $ with $ m \not= 0 $, then we say that $ c $ respects the variable convention if, and only if, 
\begin{itemize}
\item $ c_1, \ldots, c_m $ respect the variable convention~;
\item and the variables $ x_1, \ldots, x_m $ are not bound in $ t $.
\end{itemize}
\end{itemize}
For any state $s = (c_0, \ldots, c_q)$, we say that $ s $ \emph{respects the variable convention} if, and only if, $ c_0, \ldots, c_q $ respect the variable convention.

We denote by $\mathbb{S}$ the set of the states that respect the variable convention.
\end{definition}

First, we present the execution of a state (that respects the variable convention). It consists in updating a closure $ (t, e) $ and the stack. If $ t $ is an application $ (v)u $, then we push the closure $ (u, e) $ on the top of the stack and the current closure is now $ (v, e) $. If $ t $ is an abstraction, then a closure is popped and a new environment is created. If $ t $ is a variable, then the current closure is now the value of the variable of the environment. The partial map $s \succ_\mathbb{S} s'$ (defined below) defines formally the transition from a state to another state.

\begin{definition}\label{definition : aux Krivine head}
We define a partial map from $\mathbb{S}$ to $\mathbb{S}$: for any $s, s' \in \mathbb{S}$, the notation $s \succ_\mathbb{S} s'$ will mean that the map assigns $s'$ to $s$. The value of the map at $s$ is defined as follows:
\begin{itemize}
\item if $s=(x, e) . \pi$ and $x \in \textsf{dom}(e)$, then $s \succ_\mathbb{S} e(x) . \pi$;
\item if $s=(x, e) . \pi$ and $x \in \mathcal{V}$ and $x \notin \textsf{dom}(e)$, then the function is not defined at $s$;
\item if $s=(\lambda x. u, e) . (c. \pi')$, then $s \succ_\mathbb{S} (u, \{ (x, c) \} \cup e) . \pi'$;
\item if $s=(\lambda x. u, e) . \epsilon$, then the function is not defined at $s$;
\item if $s=((v)u, e) .\pi$, then $s \succ_\mathbb{S} (v, e) . ((u, e) . \pi)$.
\end{itemize}
\end{definition}

Note that in the case where the current subterm is an abstraction and the stack is empty, the machine stops: it does not reduce under lambda abstractions. That is why we slightly modify this machine in the following subsection.

\subsection[A machine computing the principal head normal form]{A machine computing the principal head normal form}\label{subsection:head}

Now, the machine has to reduce under lambda abstractions and, in Subsection \ref{subsection:beta}, the machine will have to compute the arguments of the head variable. So, we extend the machine so that it performs the reduction of elements of $ \mathcal{K} $, where $ \mathcal{K} = \bigcup_{n \in \mathbb{N}} \mathcal{K}_n $ with
\begin{itemize}
\item $ \mathcal{H}_0 = \mathcal{V}$ and $ \mathcal{K}_0 = \mathbb{S} $~;
\item $ \mathcal{H}_{n+1} = \mathcal{V} \cup \{ (v)u \: / \: v \in \mathcal{H}_n \textrm{ and } u \in \Lambda \cup \mathcal{K}_n \} $ and \\
$ \mathcal{K}_{n+1} = \mathbb{S} \cup \mathcal{H}_n \cup \{ \lambda y. k \: / \; y \in \mathcal{V} \textrm{ and } k \in \mathcal{K}_n \} . $
\end{itemize}
Set $ \mathcal{H} = \bigcup_{n \in \mathbb{N}} \mathcal{H}_n $. We have $ \mathcal{K} = \mathbb{S} \cup \mathcal{H} \cup \bigcup_{n \in \mathbb{N}} \{ \lambda x. k \: / \: x \in \mathcal{V} \textrm{ and } k \in \mathcal{K}_n \} $.

\begin{rem}
We have
\begin{itemize}
\item $ \mathcal{H} = \{ (x)t_1 \ldots t_p \; / \; p \in \mathbb{N}, \; x \in \mathcal{V}, \; t_1, \ldots, t_p \in \Lambda \cup \mathcal{K} \} $
\item hence any element of $ \mathcal{K} $ can be written as either
$$ \lambda x_1. \ldots \lambda x_m. s \textrm{ with $ m \in \mathbb{N} $, $ x_1, \ldots, x_m \in \mathcal{V} $ and $ s \in \mathbb{S} $} $$
or 
$$ \lambda x_1. \ldots \lambda x_m. (x)t_1 \ldots t_p \textrm{ with $ m, p \in \mathbb{N} $, $ x_1, \ldots, x_m \in \mathcal{V} $ and $ t_1, \ldots, t_p \in \mathcal{K} \cup \Lambda $.} $$
\end{itemize}
\end{rem}

For any $ k \in \mathcal{K} $, we denote by $ \textsf{d}(k) $ the least integer $ p $ such that $ k \in \mathcal{K}_p $.

We extend the definition of $ \overline{s} $ for $ s \in \mathbb{S} $ to $ \overline{k} $ for $ k \in \mathcal{K} $. For that, we set $ \overline{t} = t $ if $ t \in \Lambda $. This definition is by induction on $ \textsf{d}(k) $:
\begin{itemize}
\item if $ \textsf{d}(k)=0 $, then $ k \in \mathbb{S} $ and thus $ \overline{k} $ is already defined;
\item if $ k \in \mathcal{H} $, then there are two cases:
\begin{itemize}
\item if $ k \in \mathcal{V} $, then $ \overline{k} $ is already defined (it is $ k $)~;
\item else, $ k = (v)u $ and we set $ \overline{k} = (\overline{v}) \overline{u} $~;
\end{itemize}
\item if $ k = \lambda x. k_0 $, then $ \overline{k} = \lambda x. \overline{k_0} $.
\end{itemize}

\begin{definition}\label{Def Krivine head}
We define a partial map from $\mathcal{K}$ to $\mathcal{K}$: for any $k, k' \in \mathcal{K}$, the notation $k \succ_h k'$ will mean that the map assigns $k'$ to $k$. The value of the map at $k$ is defined, by induction on $\textsf{d}(k)$, as follows:
\begin{itemize}
\item if $k \in \mathbb{S}$ and $k \succ_\mathbb{S} s'$, then $k \succ_h s'$;
\item if $k = ((x, e), c_1, \ldots, c_q) \in \mathbb{S}$ and $x \in \mathcal{V}$ and $x \notin \textsf{dom}(e)$, then $k \succ_h (x)\overline{c_1} \ldots \overline{c_q}$;
\item if $k = (\lambda x. u, e) . \epsilon \in \mathbb{S}$, then $k \succ_h \lambda x.((u, e) . \epsilon)$;
\item if $k \in \mathcal{H}$, then the function is not defined at $k$;
\item if $k = \lambda y. k_0$ and $k_0 \succ_h k_0'$, then $k \succ_h \lambda y. k_0'$.
\end{itemize}
\end{definition}

A difference with the original machine is that our machine reduces under lambda abstractions.

We denote by $ {\succ_h}^\ast $ the reflexive transitive closure of $ \succ_h $. For any $ k \in \mathcal{K} $, $ k $ is said to be a \emph{Krivine normal form} if for any $ k' \in \mathcal{K} $, we do not have $ k \succ_h k' $.

\begin{definition}\label{definition : l}
For any $ k_0 \in \mathcal{K} $, we define $ l_h(k_0) \in \mathbb{N} \cup \{ \infty \} $ as follows: if there exist $ k_1, \ldots, k_n \in
\mathcal{K} $ such that $ k_i \succ_h k_{i+1} $ for $ 0 \leq i \leq n-1 $ and $ k_n $ is a Krivine normal form, then we set $ l_h(k_0) =n $, else we set $ l_h(k_0) = \infty $.
\end{definition}

\begin{prop}\label{prop : f.n. de Krivine => f.n. de tete}
For any $s \in \mathbb{S}$, for any $ k' \in \mathcal{K} $, if $ s {\succ_h}^\ast k' $ and $ k' $ is a Krivine normal form, then $ k' $ is a $\lambda$-term in head normal form.
\end{prop}

\begin{proof}
By induction on $ l_h(s) $.

The base case is trivial, because we never have $ l_h(s) =0 $.

The inductive step is divided into five cases.
\begin{itemize}
\item If $ s = ((x, e), c_1, \ldots, c_q) $, $ x \in \mathcal{V} $ and $ x \notin \textrm{dom}(e) $, then $ s \succ_h (x) \overline{c_1} \ldots \overline{c_q} $. But $ (x) \overline{c_1} \ldots \overline{c_q} $ is a Krivine normal form and $ (x) \overline{c_1} \ldots \overline{c_q} $ is a $\lambda$-term in head normal form.
\item If $ s = (\lambda x. u, e) . \epsilon$, then $ k' = \lambda x. k'' $ with $(u, e) . \epsilon {\succ_h}^\ast k'' $. Now, by induction hypothesis, $ k'' $ is a $\lambda$-term in head normal form, hence $ k' $ too is a $\lambda$-term in head normal form.
\item If $ s = ((x, e), c_1, \ldots, c_q) $, $ x \in \mathcal{V} $ and $ x \in \textrm{dom}(e) $, then $ s \succ_h (e(x), \pi) $. Now, $ e(x) . \pi {\succ_h}^\ast k' $, hence, by induction hypothesis, $ k' $ is a $\lambda$-term in head normal form.
\item If $ s = (\lambda x. u, e) . (c . \pi)$, then $ s \succ_h (u, \{ (x, c) \} \cup e) . \pi $. Now, $ (u, \{ (x, c) \} \cup e) . \pi \succ_h k' $, hence, by induction hypothesis, $ k' $ is a $\lambda$-term in head normal form.
\item If $ s = ((v)u, e) . \pi $, then $ s \succ_h (v, e) . ((u, e) . \pi) $. Now, $ (v, e) . ((u, e) . \pi) {\succ_h}^\ast k' $, hence, by induction hypothesis, $ k' $ is a $\lambda$-term in head normal form.
\end{itemize}
\end{proof}

\begin{ex}\label{example : (I)I}
Set $ s = (((\lambda x. (x)x) \lambda y. y, \emptyset), \epsilon) $. We have $ l_h(s) = 9 $:
\begin{eqnarray*}
s & \succ_h & ((\lambda x. (x)x, \emptyset), (\lambda y. y, \emptyset)) \allowdisplaybreaks \\
& \succ_h & ((x)x, \{ (x, (\lambda y. y, \emptyset)) \}) . \epsilon \allowdisplaybreaks \\
& \succ_h & ((x, \{ (x, (\lambda y. y, \emptyset)) \}) , (x, \{ (x, (\lambda y. y, \emptyset)) \})) \allowdisplaybreaks \\
& \succ_h & ((\lambda y. y, \emptyset), (x, \{ (x, (\lambda y. y, \emptyset)) \}))  \allowdisplaybreaks \\
& \succ_h & (y, \{ (y, (x, \{ (x, (\lambda y. y, \emptyset)) \})   ) \}) . \epsilon \allowdisplaybreaks \\
& \succ_h & (x, \{ (x, (\lambda y. y, \emptyset)) \} ) . \epsilon \allowdisplaybreaks \\
& \succ_h & (\lambda y. y, \emptyset) . \epsilon \allowdisplaybreaks \\
& \succ_h & \lambda y. ((y, \emptyset) . \epsilon) \allowdisplaybreaks \\
& \succ_h & \lambda y. y
\end{eqnarray*}
We present the same computation in a more descriptive way in Figure \ref{figure:example}.
\end{ex}

\begin{figure}
\centering
$\begin{array}{|c||c||c|c|c|}
\hline
& \textrm{output} & \textrm{current subterm} & \textrm{environment} & \textrm{stack} \\
\hline
& & (\lambda x. (x)x) \lambda y. y & \emptyset & \epsilon  \\
1 & & \lambda x. (x)x & \emptyset & (\lambda y. y, \emptyset) \\
2 & & (x)x & \{ x \mapsto (\lambda y. y, \emptyset) \} & \epsilon \\
3 & & x & \{ x \mapsto (\lambda y. y, \emptyset) \} & (x, \{ x \mapsto (\lambda y. y, \emptyset) \} ) \\
4 & & \lambda y. y & \emptyset & (x, \{ x \mapsto (\lambda y. y, \emptyset) \} ) \\
5 & & y & \{ y \mapsto (x, \{ x \mapsto (\lambda y. y, \emptyset) \} ) \} & \epsilon \\
6 & & x & \{ x \mapsto (\lambda y. y, \emptyset) \}  & \epsilon \\
7 & & \lambda y. y & \emptyset  & \epsilon \\
8 & \lambda y. & y & \emptyset & \epsilon \\
9 & \lambda y. y & & &\\
\hline
\end{array}$
\caption{Example of computation of the principal head normal form}\label{figure:example}
\end{figure}

\begin{lemme}\label{lemme : Krivine simule head}
For any $ k, k'\in \mathcal{K} $, if $ k \succ_h k' $, then $ \overline{k} \rightarrow_h \overline{k'} $, where $ \rightarrow_h $ is the reflexive closure of the head reduction.
\end{lemme}

\begin{proof}
There are two cases.
\begin{itemize}
\item If $ k \in \mathbb{S} $, then there are five cases.
\begin{itemize}
\item If $k = ((x, e), c_1, \ldots, c_q)$, $ x \in \mathcal{V} $ and $ x \notin \textrm{dom}(e) $, then $ \overline{k} = (x) \overline{c_1} \ldots \overline{c_q} $ and $ \overline{k'} = \overline{(x) c_1 \ldots c_q} = (x) \overline{c_1} \ldots \overline{c_q} $: we have $ \overline{k} = \overline{k'} $.
\item If $k = (\lambda x. u, e) . \epsilon $, then $ \overline{k} = (\lambda x. u)[e] = \lambda x. u[e] $ (because $k$ respects the variable convention) and $ \overline{k'} = \overline{\lambda x. ((u, e), \epsilon)} = \lambda x. u[e] $: we have $ \overline{k} = \overline{k'} $.
\item If $k = ((x, e), c_1, \ldots, c_q) $, $ x \in \mathcal{V} $ and $ x \in \textrm{dom}(e) $, then $ \overline{k} = \overline{e(x)} \overline{c_1} \ldots \overline{c_q} $ and $ \overline{k'} = \overline{(e(x), (c_1, \ldots, c_q))} = \overline{e(x)} \overline{c_1} \ldots \overline{c_q} $: we have $ \overline{k} = \overline{k'} $.
\item If $k = ((\lambda x. u, e), c_0, \ldots, c_q)$, then $ \overline{k} = ((\lambda x. u)[e]) \overline{c_0} \ldots \overline{c_q} = (\lambda x. u[e]) \overline{c_0} \ldots \overline{c_q} $ (since $k$ respects the variable convention) and $ \overline{k'} = (\overline{(u, \{ (x, c_0) \} \cup e)}) \overline{c_1} \ldots \overline{c_q} $. Now, $ \overline{k} $ reduces in a single head reduction step to $ \overline{k'} $.
\item If $k = (((v)u, e), c_1 \ldots c_q) $, then $ \overline{k} = (((v)u)[e]) \overline{c_1} \ldots \overline{c_q} = (v[e])u[e] \overline{c_1} \ldots \overline{c_q} $ and $ \overline{k'} = \overline{((v, e), (u, e), c_1, \ldots, c_q)} = (v[e])u[e] \overline{c_1} \ldots \overline{c_q} $: we have $ \overline{k} = \overline{k'} $.
\end{itemize}
\item Else, $ k = \lambda y. k_0 $~; then $ \overline{k} = \lambda y. \overline{k_0} $ and $ \overline{k'} = \overline{\lambda y. k_0'} = \lambda y. \overline{k_0'} $ with $ k_0 \succ_h k'_0 $: we have $ \overline{k_0} \rightarrow_h \overline{k_0'} $, hence $ \overline{k} \rightarrow_h \overline{k'} $. 
\end{itemize}
\end{proof}

\begin{theorem}\label{th Krivine fini => head}
For any $ k \in \mathcal{K} $, if $ l_h(k) $ is finite, then $ \overline{k} $ is head normalizable.
\end{theorem}

\begin{proof}
By induction on $ l_h(k) $.

If $ l_h(k) = 0 $, then $ k \in \mathcal{H} $, hence $ k $ can be written as $ (x)t_1 \ldots t_p $ and thus $ \overline{k} $ can be written $ (x)\overline{t_1} \ldots \overline{t_p} $: it is a head normal form. Else, apply Lemma \ref{lemme : Krivine simule head}.
\end{proof}

For any head normalizable $\lambda$-term $ t $, we denote by $ \textsf{h}(t) $ the number of head reductions of $ t $.

\begin{theorem}\label{theorem : head-normalizable => l fini}
For any $ s = (t, e) . \pi \in \mathbb{S}$, if $ \overline{s} $ is head normalizable, then $ l_h(s) $ is finite.
\end{theorem}

\begin{proof}
We prove, by noetherian induction on $\mathbb{N} \times \mathbb{N} \times \Lambda$ lexically ordered, that for any $(h, d, t) \in \mathbb{N} \times \mathbb{N} \times \Lambda$, for any $s = (t, e) . \pi$ such that $\textsf{h}(s) = h$ and $\textsf{d}(s) =d$, if $s$ is head-normalizable, then $l_h(s)$ is finite.

If $ \textsf{h}(\overline{s}) = 0 $, $ \textsf{d}(e)=0 $ and $ t \in \mathcal{V} $, then we have $ l_h(s) = 1 $.

Else, there are five cases.
\begin{itemize}
\item In the case where $t \in \textrm{dom}(e) $, we have $ s \succ_h e(t) . \pi$. Set $ s' = e(t) . \pi$ and $ e(t) = (t', e')$. We have $ \overline{s} = \overline{s'} $ and $ \textsf{d}(e') < \textsf{d}(e) $, hence we can apply the induction hypothesis: $ l_h(s') $ is finite and thus $ l_h(s) = l_h(s') + 1 $ is finite.
\item In the case where $t \in \mathcal{V} $ and $ t \notin \textrm{dom}(e) $, we have $ l_h(s) = 1 $.
\item In the case where $ t = (v) u $, we have $s \succ_h (v, e) . ((u, e) . \pi) $. Set $ s' = (v, e) . ((u, e) . \pi) $. We have $ \overline{s'} = \overline{s} $ and thus we can apply the induction hypothesis: $ l_h(s') $ is finite and thus $ l_h(s) = l_h(s') + 1 $ is finite.
\item In the case where $ t = \lambda x. u $ and $ \pi = \epsilon $, we have $ s \succ_h \lambda x. ((u, e) . \epsilon) $. Set $ s'= (u, e) . \epsilon $. Since $ s $ respects the variable convention, we have $ \overline{s} = \lambda x. u[e] = \lambda x. \overline{s'} $. We have $ \textsf{h}(\overline{s'}) = \textsf{h}(\overline{s}) $, hence we can apply the induction hypothesis: $ l_h(s') $ is finite and thus $ l_h(s) = l_h(s') + 1 $ is finite.
\item In the case where $ t = \lambda x. u $ and $ \pi = c . \pi' $, we have $ s \succ_h (u, \{ (x, c) \} \cup e) . \pi $. Set $ s'= (u, \{ (x, c) \} \cup e) .\pi$. We have $ \textsf{h}(\overline{s'}) < \textsf{h}(\overline{s}) $, hence we can apply the induction hypothesis: $ l_h(s') $ is finite and thus $ l_h(s) = l_h(s') + 1 $ is finite.
\end{itemize}
\end{proof}

We recall that if a $\lambda$-term $t$ has a head-normal form, then the last term of the terminating head reduction of $t$ is called \emph{the principal head normal form of $t$} (see \cite{Barendregt}). Proposition \ref{prop : f.n. de Krivine => f.n. de tete}, Lemma \ref{lemme : Krivine simule head} and Theorem \ref{theorem : head-normalizable => l fini} show that for any head normalizable $\lambda$-term $ t$ having $t'$ as principal head normal form, we have $ (t, \emptyset) . \epsilon {\succ_h}^\ast t' $ and $ t' $ is a Krivine head normal form.

\subsection{A machine computing the $ \beta $-normal form}\label{subsection:beta}

We now slightly modify the machine so as to compute the $ \beta $-normal form of any normalizable $\lambda$-term.

\begin{definition}\label{Def Krivine normal}
We define a partial map from $\mathcal{K}$ to $\mathcal{K}$: for any $k, k' \in \mathcal{K}$, the notation $k \succ_\beta k'$ will mean that the map assigns $k'$ to $k$. The value of the map at $k$ is defined, by induction on $\textsf{d}(k)$, as follows:
$$ k \mapsto \left\lbrace \begin{array}{ll} s' & \textrm{if $k \in \mathbb{S}$ and $k \succ_\mathbb{S} s'$}\\
(x) (c_1 . \epsilon) \ldots (c_q . \epsilon) & \textrm{if $k = ((x, e), c_1, \ldots, c_q) \in \mathbb{S}$, $x \in \mathcal{V}$ and $x \notin \textsf{dom}(e)$}\\
\lambda x.((u, e) . \epsilon) & \textrm{if $k = ((\lambda x. u, e) . \epsilon) \in \mathbb{S}$}\\
\textrm{not defined} & \textrm{if $k \in \mathcal{V}$}\\
(v')u & \textrm{if $k = (v)u$ and $v \succ_\beta v'$}\\
(x)u' & \textrm{if $k = (x)u$ with $x \in \mathcal{V}$ and $u \succ_\beta u'$}\\
\lambda y. k_0' & \textrm{if $k = \lambda y. k_0$ and $k_0 \succ_\beta k_0'$} \end{array} \right.$$
\end{definition}

Let us compare Definition \ref{Def Krivine normal} with Definition \ref{Def Krivine head}. The difference is in the case where the current subterm of a state is a variable and where this variable has no value in the environment: the first machine stops, the second machine continues to compute every argument of the variable.

The function $ l_\beta $ is defined as $ l_h $ (see Definition \ref{definition : l}), but for this new machine. 

For any normalizable $\lambda$-term $ t $, we denote by $ n(t) $ the number of steps leading from $t$ to its normal form following the leftmost reduction strategy.

\begin{theorem}\label{theorem : normalizable => l' finite}
For any $s = (t, e) . \pi \in \mathbb{S}$, if $\overline{s} $ is normalizable, then $ l_\beta(s) $ is finite.
\end{theorem}

\begin{proof}
We prove, by noetherian induction on $\mathbb{N} \times \mathbb{N} \times \Lambda$ lexicographically ordered, that for any $(h, d, t) \in \mathbb{N} \times \mathbb{N} \times \Lambda$, for any $s = (t, e) . \pi$ such that $\textsf{h}(s) = h$ and $\textsf{d}(s) =d$, if $s$ is head-normalizable, then $l_h(s)$ is finite.

If $ n(\overline{s}) = 0 $, $ \overline{s} \in \mathcal{V} $, $ \textsf{d}(e)=0 $ and $ t \in \mathcal{V} $, then we have $ l_\beta(s) = 1 $.

Else, there are five cases.
\begin{itemize}
\item In the case where $ t \in \mathcal{V} \cap \textrm{dom}(e) $, we have $ s \succ_\beta (e(t), \pi) $. Set $ s' = (e(t), \pi) $ and $ e(t) = (t', e') $. We have $ \overline{s} = \overline{s'} $ and $ \textsf{d}(e') < \textsf{d}(e) $, hence we can apply the induction hypothesis: $ l_\beta(s') $ is finite and thus $ l_\beta(s) = l_\beta(s') + 1 $ is finite.
\item In the case where $ t \in \mathcal{V} $ and $ t \notin \textrm{dom}(e) $, set $ \pi = (c_1, \ldots, c_q) $. For any $ k \in \{ 1, \ldots, q \} $, we have $ n(\overline{c_k}) \leq n(\overline{s}) $ and $\overline{c_k} < \overline{s} $, hence we can apply the induction hypothesis on $ c_k $: for any $ k \in \{ 1, \ldots, q \} $, $ l_\beta(c_k) $ is finite, hence $ l_\beta(s) = \sum_{k=1}^q l_\beta(c_k) + 1 $ is finite too.
\item In the case where $ t = (v) u $, we have $ s \succ_\beta (v, e) . ((u, e) . \pi) $. Set $ s' = (v, e) . ((u, e) . \pi) $. We have $ \overline{s'} = \overline{s} $, hence we can apply the induction hypothesis: $ l_\beta(s') $ is finite and thus $ l_\beta(s) = l_\beta(s') + 1 $ is finite.
\item In the case where $ t = \lambda x. u $ and $ \pi = \epsilon $, we have $ s \succ_\beta \lambda x. ((u, e) . \epsilon) $. Set $ s'= (u, e) . \epsilon $. Since $ s $ respects the variable convention, we have $ \overline{s} = \lambda x. u[e] = \lambda x. \overline{s'} $. We have $ n(\overline{s'}) = n(\overline{s}) $, hence we can apply the induction hypothesis: $ l_\beta(s') $ is finite and thus $ l_\beta(s) = l_\beta(s') + 1 $ is finite.
\item In the case where $ t = \lambda x. u $ and $ \pi = c . \pi' $, we have $ s \succ_\beta (u, \{ (x, c) \} \cup e) . \pi $. Set $ s'= (u, \{ (x, c) \} \cup e) . \pi $. We have $ n(\overline{s'}) < n(\overline{s}) $, hence we can apply the induction hypothesis: $ l_\beta(s') $ is finite and thus $ l_\beta(s) = l_\beta(s') + 1 $ is finite.
\end{itemize}
\end{proof}

%% file: execution_time_semantics.tex
\section{A non-uniform semantics of $\lambda$-calculus}\label{section : semantics}

We define here the semantics allowing to measure execution time. We have in mind the following philosophy: the semantics of the untyped $\lambda$-calculus come from the semantics of the simply typed $\lambda$-calculus and any semantics of linear logic induces a semantics of the simply typed $\lambda$-calculus. So, we start from a semantics $\mathfrak{M}$ of linear logic (Subsection \ref{subsection:linear logic}), then we present the induced semantics $\Lambda(\mathfrak{M})$ of the simply typed $\lambda$-calculus (Subsection \ref{subsection:typed}) and lastly the semantics of the untyped $\lambda$-calculus that we consider (Subsection \ref{subsection:untyped}). This semantics is \emph{non-uniform} in the sense that the interpretation of a function contains information abouts its behaviour on arguments whose value can change during the computation: in Subsection \ref{subsection:non uniformity}, we give an example illustrating this point.

The first works tackling the problem of giving a general categorical definition of a denotational semantics of linear logic are those of Lafont \cite{lafont} and of Seely \cite{seely}. 
As for the works of Benton, Bierman, Hyland and de Paiva, \cite{BBPH}, \cite{Bierman_thesis} and \cite{Bierman1995}, they led to the following axiomatic: a categorical model of the multiplicative exponential fragment of intuitionistic linear logic (IMELL) is a quadruple $ (\mathcal{C}, \mathcal{L}, c, w) $ such that
\begin{itemize}
\item $ \mathcal{C} = (\mathbb{C}, \otimes, I, \alpha, \lambda, \rho, \gamma) $ is a closed symmetric monoidal category;
\item $ \mathcal{L} = ((T, \textsf{m}, \textsf{n}), \delta, d) $ is a symmetric monoidal comonad on $ \mathcal{C} $;
\item $ c $ is a monoidal natural transformation from $ (T, \textsf{m}, \textsf{n}) $ to $ \otimes \circ \Delta_\mathcal{C} \circ (T, \textsf{m}, \textsf{n}) $ and $ w $ is a monoidal natural transformation from $ (T, \textsf{m}, \textsf{n}) $ to $ \ast_\mathcal{C} $ such that
\begin{itemize}
\item for any object $ A $ of $ \mathbb{C} $, $ ((T(A), \delta_A), c_A, w_A) $ is a cocommutative comonoid in $ (\mathbb{C}^\mathbb{T}, \otimes^\mathbb{T}, (I, \textsf{n}), \alpha, \lambda, \rho) $
\item and for any $ f \in \mathbb{C}^\mathbb{T}[(T(A), \delta_A), (T(B), \delta_B)] $, $ f $ is a comonoid morphism,
\end{itemize}
where $ \mathbb{T} $ is the comonad $ (T, \delta, d) $ on $ \mathbb{C} $, $ \mathbb{C}^\mathbb{T} $ is the category of $ \mathbb{T} $-coalgebras, $ \Delta_\mathcal{C} $ is the diagonal monoidal functor from $ \mathcal{C} $ to $ \mathcal{C} \times \mathcal{C} $ and $ \ast_\mathcal{C} $ is the monoidal functor that sends any arrow to $ id_I $.
\end{itemize}

Given a categorical model $ \mathfrak{M} = (\mathcal{C}, \mathcal{L}, c, w) $ of IMELL with $ \mathcal{C} = (\mathbb{C}, \otimes, I, \alpha, \lambda, \rho, \gamma) $ and $ \mathcal{L} = ((T, \textsf{m}, \textsf{n}), \delta, d) $, we can define a cartesian closed category $ \Lambda(\mathfrak{M}) $ such that
\begin{itemize}
\item objects are finite sequences of objects of $ \mathbb{C} $
\item and arrows $ ( A_1, \ldots, A_m ) \rightarrow ( B_1, \ldots, B_p ) $ are the sequences $ ( f_1, \ldots, f_p ) $ such that every $f_k$ is an arrow $ \bigotimes_{j=1}^m T(A_j) \rightarrow B_k $ in $ \mathbb{C} $.
\end{itemize}

Hence we can interpret simply typed $\lambda$-calculus in the category $ \Lambda(\mathfrak{M}) $. 
This category is (weakly) equivalent\footnote{A category $\mathbb{C}$ is said to be \emph{weakly equivalent} to a category $\mathbb{D}$ if there exists a functor $F:\mathbb{C} \rightarrow \mathbb{D}$ full and faithful such that every object $D$ of $\mathbb{D}$ is isomorphic to $F(C)$ for some object $C$ of $\mathbb{C}$.} to a full subcategory of $ (T, \delta, d) $-coalgebras exhibited by Hyland. If the category $ \mathbb{C} $ is cartesian, then the categories $ \Lambda(\mathfrak{M}) $ and the Kleisli category of the comonad $(T, \delta, d)$ are (strongly) equivalent\footnote{A category $\mathbb{C}$ is said to be \emph{strongly equivalent} to a category $\mathbb{D}$ if there are functors $F:\mathbb{C} \rightarrow \mathbb{D}$ and $G:\mathbb{D} \rightarrow \mathbb{C}$ and natural isomorphisms $G \circ F \cong id_\mathbb{C}$ and $F \circ G \cong id_\mathbb{D}$.}. See \cite{these} for a full exposition.

Below, we describe completely the category $\Lambda(\mathfrak{M})$ (with its composition operation and its identities) only for the particular case that we consider in this paper.

\subsection{A relational model of linear logic}\label{subsection:linear logic}

The category of sets and relations is denoted by $\mathbf{Rel}$ and its composition operation by $\circ$. The functor $ T $ from $ \mathbf{Rel} $ to $ \mathbf{Rel} $ is defined by setting 
\begin{itemize}
\item for any object $ A $ of $ \mathbf{Rel} $, $ T(A) = \mathcal{M}_{\textrm{fin}}(A) $;
\item and, for any $ f \in \mathbf{Rel}(A, B) $, $ T(f) \in \mathbf{Rel}(T(A), T(B)) $ defined by 
$$ T(f) = \{ ([\alpha_1, \ldots, \alpha_n], [\beta_1, \ldots, \beta_n])\: / \: n \in \mathbb{N} \textrm{ and } (\alpha_1, \beta_1), \ldots, (\alpha_n, \beta_n) \in f \} . $$
\end{itemize}
The natural transformation $ d $ from $ T $ to the identity functor of $ \mathbf{Rel} $ is defined by setting $ d_A = \{ ([\alpha], \alpha) \: / \: \alpha \in A \} $ and the natural transformation $ \delta $ from $ T $ to $ T \circ T $ by setting $ \delta_A = \{ (a_1 + \ldots + a_n, [a_1, \ldots, a_n]) \: / \: n \in \mathbb{N} \textrm{ and } a_1, \ldots, a_n \in T(A) \} . $ It is easy to show that $ (T, \delta, d) $ is a comonad on $ \mathbf{Rel} $. It is well-known that this comonad can be provided with a structure $ \mathfrak{M} $ that is a denotational semantics of (I)MELL. 

This denotational semantics gives rise to a cartesian closed category $\Lambda(\mathfrak{M})$.

\subsection{Interpreting simply typed $\lambda$-terms}\label{subsection:typed}

We give the complete description of the category $\Lambda(\mathfrak{M})$ induced by the denotational semantics $\mathfrak{M}$ of (I)MELL evoked in the previous subsection:
\begin{itemize}
\item objects are finite sequences of sets;
\item arrows $( A_1, \ldots, A_m ) \rightarrow ( B_1, \ldots, B_n )$ are the sequences $( f_1, \ldots, f_n )$ such that every $f_i$ is a subset of $(\prod_{j=1}^m \mathcal{M}_{\textrm{fin}}(A_j)) \times B_i$ with the convention $(\prod_{j=1}^m \mathcal{M}_{\textrm{fin}}(A_j)) \times B_i = B_i$ if $m=0$;
\item if $( f_1, \ldots, f_p )$ is an arrow $( A_1, \ldots, A_m ) \rightarrow ( B_1, \ldots, B_p ) $ and $( g_1, \ldots, g_q )$ is an arrow $ ( B_1, \ldots, B_p ) \rightarrow ( C_1, \ldots, C_q )$, then $( g_1, \ldots, g_q ) \circ_{\Lambda(\mathfrak{M})} ( f_1, \ldots, f_p )$ is the arrow $ ( h_1, \ldots, h_q ) : (A_1, \ldots, A_m) \rightarrow (C_1, \ldots, C_q)$, where $h_l$ is
$$ \bigcup_{n_1, \ldots, n_p \in \mathbb{N}} \left\lbrace \begin{array}{l} ((\sum_{k=1}^p \sum_{i=1}^{n_k} a_1^{i, k}, \ldots, \sum_{k=1}^p \sum_{i=1}^{n_k} a_m^{i, k}), \gamma) \: / \: \\
(\forall j \in \{ 1, \ldots, m \}) (\forall k \in \{ 1, \ldots, p \})(a_j^{i, k})_{1 \leq i \leq n_k} \in {(\mathcal{M}_{\textrm{fin}}(A_j))}^{n_k} \\
\textrm{and } (\exists \beta_1^1, \ldots, \beta_1^{n_1} \in B_1) \ldots (\exists \beta_p^1, \ldots, \beta_p^{n_p} \in B_p) \\
\begin{array}{ll}
& ((([\beta_1^1, \ldots, \beta_1^{n_1}], \ldots, [\beta_p^1, \ldots, \beta_p^{n_p}]), \gamma) \in g_l \textrm{ and }\\
& (\forall k \in \{ 1, \ldots, p \}) (\forall i \in \{ 1, \ldots, n_k \}) ((a_1^{i, k}, \ldots, a_m^{i, k}), \beta_k^i) \in f_k) \end{array}
\end{array} \right\rbrace $$
$ \textrm{ for } 1 \leq l \leq q $, with the conventions 
$$((a_1, \ldots, a_m), \gamma)=\gamma \textrm{ and } (\prod_{j=1}^m \mathcal{M}_{\textrm{fin}}(A_j)) \times C_l=C_l \textrm{ if } m=0 ; $$
\item the identity of $( A_1, \ldots, A_m )$ is $( d^1, \ldots, d^m )$ with 
$$d^j = \{ ((\underbrace{[], \ldots, []}_{j-1 \textrm{ times}}, [\alpha], \underbrace{[], \ldots, []}_{m-j \textrm{ times}}), \alpha) \: / \: \alpha \in A_j \} . $$
\end{itemize}

\begin{prop}\label{prop : CCC}
The category $\Lambda(\mathfrak{M})$ has the following cartesian closed structure 
$$(\Lambda(\mathfrak{M}), 1, !, \&, \pi^1, \pi^2, ( \cdot, \cdot )_{\mathfrak{M}}, \Rightarrow, \Lambda, \textrm{ev}): $$
\begin{itemize}
\item the terminal object $1$ is the empty sequence $( )$;
\item if $B^1 = ( B_1, \ldots, B_p ) $ and $B^2 = ( B_{p+1}, \ldots, B_{p+q} ) $ are two sequences of sets, then $ B^1 \& B^2 $ is the sequence $ ( B_1, \ldots, B_{p+q} ) $;
\item if $B^1 = ( B_1, \ldots, B_p ) $ and $B^2 = ( B_{p+1}, \ldots, B_{p+q} ) $ are two sequences of sets, then 
$$ \pi_{B^1, B^2}^1 = ( d^1, \ldots, d^p ) : B^1 \& B^2 \rightarrow B^1 \textrm{ in } \Lambda(\mathfrak{M}) $$
and
$$ \pi_{B^1, B^2}^2 = ( d^{p+1}, \ldots, d^{p+q} ) : B^1 \& B^2 \rightarrow B^2 \textrm{ in } \Lambda(\mathfrak{M}) $$
with
$$d^k = \{ ((\underbrace{[], \ldots, []}_{k-1 \textrm{ times}}, [\beta], \underbrace{[], \ldots, []}_{p+q-k \textrm{ times}}), \beta) \: / \: \beta \in B_k \} ; $$
\item if $f^1 = ( f_1, \ldots, f_p ) : C \rightarrow A^1$ and $f^2 = ( f_{p+1}, \ldots, f_{p+q} ) : C \rightarrow A^2$ in $\Lambda(\mathfrak{M})$, then $( f^1, f^2 )_{\mathfrak{M}} = ( f_1, \ldots, f_{p+q} ) : C \rightarrow A^1 \& A^2 $;
\end{itemize}
and
\begin{itemize}
\item $( A_1, \ldots, A_m ) \Rightarrow ( C_1, \ldots, C_q )$ is defined by induction on $m$:
\begin{itemize}
\item $( ) \Rightarrow ( C_1, \ldots, C_q ) = ( C_1, \ldots, C_q )$
\item 
\begin{eqnarray*}
& & ( A_1, \ldots, A_{m+1} ) \Rightarrow ( C_1, \ldots, C_q ) \\
& = & ( ( A_1, \ldots, A_m ) \Rightarrow (\mathcal{M}_{\textrm{fin}}(A_{m+1}) \times C_1), \ldots, \\
& & \: \: \: ( A_1, \ldots, A_m ) \Rightarrow (\mathcal{M}_{\textrm{fin}}(A_{m+1}) \times C_q) ) ;
\end{eqnarray*}
\end{itemize}
\item if $h = ( h_1, \ldots, h_q ) : ( A_1, \ldots, A_m ) \& ( B_1, \ldots, B_p ) \rightarrow ( C_1, \ldots, C_q ) $, then 
$$\Lambda_{( A_1, \ldots, A_m ), ( C_1, \ldots, C_q )}^{( B_1, \ldots, B_p )}(h) : ( A_1, \ldots, A_m ) \rightarrow ( B_1 \ldots, B_p ) \Rightarrow ( C_1, \ldots, C_q ) $$ 
is defined by induction on $p$:
\begin{itemize}
\item if $p=0$, then $\Lambda_{( A_1, \ldots, A_m ), ( C_1, \ldots, C_q )}^{( B_1, \ldots, B_p )}(h)=h$;
\item if $p=1$, then there are two cases:
\begin{itemize}
\item in the case $m=0$, $\Lambda_{( A_1, \ldots, A_m ), ( C_1, \ldots, C_q )}^{( B_1, \ldots, B_p )}(h) = h $;
\item in the case $m \not=0 $, 
$$\Lambda_{( A_1, \ldots, A_m ), ( C_1, \ldots, C_q )}^{( B_1, \ldots, B_p )}(h) = ( \xi_{\prod_{j=1}^m \mathcal{M}_{\textrm{fin}}(A_j), C_1}^{\mathcal{M}_{\textrm{fin}}(B_1)}(h_1), \ldots, \xi_{\prod_{j=1}^m \mathcal{M}_{\textrm{fin}}(A_j), C_q}^{\mathcal{M}_{\textrm{fin}}(B_1)}(h_q) ) , $$
where
$$ \xi_{\prod_{j=1}^m \mathcal{M}_{\textrm{fin}}(A_j), C_l}^{\mathcal{M}_{\textrm{fin}}(B_1)}(h_l) = \{ (a, (b, \gamma)) \: ; \: ((a, b), \gamma) \in h_l \} 
 ; $$
\end{itemize}
\item if $p \geq 1$, then 
\begin{eqnarray*}
& & \Lambda_{A, ( C_1, \ldots, C_q )}^{( B_1, \ldots, B_{p+1} )}(h) \\
& = & \Lambda_{A, ( \mathcal{M}_{\textrm{fin}}(B_{p+1}) \times C_1, \ldots, \mathcal{M}_{\textrm{fin}}(B_{p+1}) \times C_q))}^{( B_1, \ldots, B_p )}(\Lambda_{( A_1, \ldots, A_m, B_1, \ldots, B_p ), ( C_1, \ldots, C_q )}^{( B_{p+1} )}(h))  ,
\end{eqnarray*}
where $ A = ( A_1, \ldots, A_m )$;
\end{itemize}
\item $\textrm{ev}_{C, B} : (B \Rightarrow C) \& B \rightarrow C $ is defined by setting 
$$\textrm{ev}_{( C_1, \ldots, C_q ), ( B_1, \ldots, B_p )} = ( \textrm{ev}_{( C_1, \ldots, C_q ), ( B_1, \ldots, B_p )}^1, \ldots, \textrm{ev}_{( C_1, \ldots, C_q ), ( B_1, \ldots, B_p )}^q )$$ where, for $1 \leq k \leq q$, 
\begin{eqnarray*}
& & \textrm{ev}_{( C_1, \ldots, C_q ), ( B_1, \ldots, B_p )}^k \\
& = & \left\lbrace \begin{array}{l} ((\underbrace{[], \ldots, []}_{k-1 \textrm{ times}}, [((b_1, \ldots, b_p), \gamma)], \underbrace{[], \ldots, []}_{q-k \textrm{ times}}, b_1, \ldots, b_p), \gamma) \: / \\
\: \: \: b_1 \in \mathcal{M}_{\textrm{fin}}(B_1), \ldots, b_p \in \mathcal{M}_{\textrm{fin}}(B_p), \gamma \in C_k \end{array} \right\rbrace .
\end{eqnarray*}
\end{itemize}
\end{prop}

\begin{proof}
By checking some computations or by applying the theorem that states that if $\mathfrak{M}$ is a denotational semantics of IMELL, then the "induced" structure 
$$(\Lambda(\mathfrak{M}), 1, !, \&, \pi^1, \pi^2, ( \cdot, \cdot )_{\mathfrak{M}}, \Rightarrow, \Lambda, \textrm{ev})$$ 
is a cartesian closed structure (see \cite{these}).
\end{proof}

\subsection{Interpreting type free $\lambda$-terms}\label{subsection:untyped}

First, we recall that if $f : D \rightarrow C$ and $g: C \rightarrow D$ are two arrows in a category $\mathbb{C}$, then \emph{$f$ is a retraction of $g$ in $ \mathbb{C} $} means that $f \circ_\mathbb{C} g = id_{C} $ (see, for instance, \cite{CWM}); it is also said that \emph{$ (g, f)$ is a retraction pair}.

With the cartesian closed structure on $\Lambda(\mathfrak{M})$, we have a semantics of the simply typed $\lambda$-calculus (see, for instance, \cite{LS}). Now, in order to have a semantics of the pure $\lambda$-calculus, it is therefore enough to have a \emph{reflexive} object $ U $ of $ \Lambda(\mathfrak{M}) $, that is to say such that
$$ (U \Rightarrow U) \lhd U , $$
that means that there exist $ s \in \Lambda(\mathfrak{M})[U \Rightarrow U, U] $ and $ r \in \Lambda(\mathfrak{M})[U, U \Rightarrow U] $ such that $ r \circ_{\Lambda(\mathfrak{M})} s $ is the identity on $ U \Rightarrow U $; in particular, $(s, r)$ is a retraction pair. We will use the following lemma for exhibiting such a retraction pair. 

\begin{lemme}\label{lemma : retraction}
Let $ h : A \rightarrow B$ be an injection between sets. Consider the arrows $g : \mathcal{M}_{\textrm{fin}}(A) \rightarrow B$ and $f : \mathcal{M}_{\textrm{fin}}(B) \rightarrow A$ of the category $\mathbf{Rel}$ defined by
\mbox{$g = \{ ([\alpha], h(\alpha)) \/ / \/ \alpha \in A \}$} and $ f = \{ ([h(\alpha)], \alpha) \/ / \/ \alpha \in A \}$.
Then $ ( g ) \in \Lambda(\mathfrak{M})(( A ), ( B )) $ and $ ( f ) $ is a retraction of $ ( g ) $ in $ \Lambda(\mathfrak{M}) $.
\end{lemme}

\begin{proof}
An easy computation shows that we have
\begin{eqnarray*}
( f ) \circ_{\Lambda(\mathfrak{M})} ( g ) & = & ( f \circ T(g) \circ \delta_A ) \\
& = & ( d_A ) .
\end{eqnarray*}
\end{proof}

If $ D $ is a set, then $ ( D ) \Rightarrow ( D ) = ( \mathcal{M}_{\textrm{fin}}(D) \times D ) $. From now on, we assume that $ D $ is a non-empty set and that $h$ is an injection from $ \mathcal{M}_{\textrm{fin}}(D) \times D $ to $ D $. Set 
$$ g = \{ ([\alpha], h(\alpha)) \: / \: \alpha \in \mathcal{M}_{\textrm{fin}}(D) \times D \} : \mathcal{M}_{\textrm{fin}}(\mathcal{M}_{\textrm{fin}}(D) \times D) \rightarrow D \textrm{ in } \mathbf{Rel} $$
and 
$$ f = \{ ([h(\alpha)], \alpha) \: / \: \alpha \in \mathcal{M}_{\textrm{fin}}(D) \times D \} : \mathcal{M}_{\textrm{fin}}(D) \rightarrow \mathcal{M}_{\textrm{fin}}(D) \times D \textrm{ in } \mathbf{Rel} . $$ 
We have 
$$ (( D ) \Rightarrow ( D )) \lhd ( D ) \text{ in the category } \Lambda(\mathfrak{M})$$
and, more precisely: \mbox{$ ( g ) \in \Lambda(\mathfrak{M})(( D ) \Rightarrow ( D ), ( D )) $} and $(f)$ is a retraction of $(g)$.

We can therefore define the interpretation of any $\lambda$-term.

\begin{definition}\label{def : [t]}
For any $\lambda$-term $ t $ possibly containing constants from $ \mathcal{P}(D) $, for any $ x_1, \ldots, x_m \in \mathcal{V} $ distinct such that $ FV(t) \subseteq \{ x_1, \ldots, x_m \} $, we define, by induction on $ t $, $ \llbracket t \rrbracket_{x_1, \ldots, x_m} \subseteq (\prod_{j=1}^m \mathcal{M}_{\textrm{fin}}(D)) \times D$:
\begin{itemize}
\item $ \llbracket x_j \rrbracket_{x_1, \ldots, x_m} = \{ ((\underbrace{[], \ldots, []}_{j-1 \textrm{ times}}, [\alpha], \underbrace{[], \ldots, []}_{m-j \textrm{ times}}), \alpha) \: / \: \alpha \in D \}$;
\item for any $ c \in \mathcal{P}(D) $, $ \llbracket c \rrbracket_{x_1, \ldots, x_m} = (\prod_{j=1}^m \mathcal{M}_{\textrm{fin}}(D)) \times c$;
\item $ \llbracket \lambda x. u \rrbracket_{x_1, \ldots, x_m} = \{ ((a_1, \ldots, a_m), h(a, \alpha)) \: / \: ((a_1, \ldots, a_m, a), \alpha) \in \llbracket u \rrbracket_{x_1, \ldots, x_m, x} \} $;
\item the value of $\llbracket (v)u \rrbracket_{x_1, \ldots, x_m}$ is 
$$\bigcup_{n \in \mathbb{N}} \bigcup_{\alpha_1, \ldots, \alpha_n \in D} \left\lbrace \begin{array}{l} ((\sum_{i=0}^n a_1^i, \ldots, \sum_{i=0}^n a_m^i), \alpha) \: / \: \\
\begin{array}{ll}
& ((a_1^0, \ldots, a_m^0), h([\alpha_1, \ldots, \alpha_n],\alpha)) \in \llbracket v \rrbracket_{x_1, \ldots, x_m} \\
\textrm{ and} & (\forall i \in \{ 1, \ldots, n \}) ((a_1^i, \ldots, a_m^i), \alpha_i) \in \llbracket u \rrbracket_{x_1, \ldots, x_m} \end{array} \end{array} \right\rbrace ; $$
\end{itemize}
with the conventions $ (\prod_{j=1}^m \mathcal{M}_{\textrm{fin}}(D)) \times D=D$ and $ ((a_1, \ldots, a_m), \alpha) = \alpha $ if $ m = 0 $.
\end{definition}

Now, we can define the interpretation of any $\lambda$-term in any environment.

\begin{definition}\label{Def : [t]_rho}
For any $ \rho \in {\mathcal{P}(D)}^\mathcal{V} $ and for any $\lambda$-term $ t $ possibly containing constants from $ \mathcal{P}(D) $ such that $ FV(t) = \{ x_1, \ldots, x_m \} $, we set 
$$ \llbracket t \rrbracket_\rho = \bigcup_{a_1 \in \mathcal{M}_{\textrm{fin}}(\rho(x_1)), \ldots, a_m \in \mathcal{M}_{\textrm{fin}}(\rho(x_m))} 
\{ \alpha \in D \; / \; ((a_1, \ldots, a_m), \alpha) \in \llbracket t \rrbracket_{x_1, \ldots, x_m}) \} . $$
\end{definition}

For any $ d_1, d_2 \in \mathcal{P}(D) $, we set 
$$ d_1 \ast d_2 = \bigcup_{a \in \mathcal{M}_{\textrm{fin}}(d_2)} \{ \alpha \in D \; / \; h(a, \alpha) \in d_1 \} . $$
We have

\begin{prop}
The triple $(\mathcal{P}(D), \ast, \llbracket - \rrbracket_-)$ is a $\lambda$-algebra.
\end{prop}

\begin{proof}
Apply our Proposition \ref{prop : CCC} and Lemma \ref{lemma : retraction}, and Theorem 5.5.6 of \cite{Barendregt}. 
\end{proof}

But the following proposition, a corollary of Proposition \ref{prop : not enough points}, states that the triple $(\mathcal{P}(D), \ast, \llbracket - \rrbracket_-)$ \emph{is not} a $\lambda$-model. We recall (see, for instance, \cite{Barendregt}), that a $\lambda$-model is a $\lambda$-algebra $(\mathcal{D}, \ast, \llbracket - \rrbracket_-)$ such that the following property, expressing the $\xi$-rule, holds: 

for any $ \rho \in \mathcal{D}^\mathcal{V} $, for any $ x \in \mathcal{V} $ and for any $\lambda$-terms $ t_1 $ and $ t_2 $, we have 
$$ ((\forall d \in \mathcal{D}) \llbracket t_1 \rrbracket_{\rho[x:=d]} = \llbracket t_2 \rrbracket_{\rho[x:=d]} \Rightarrow \llbracket \lambda x. t_1 \rrbracket_\rho = \llbracket \lambda x. t_2 \rrbracket_\rho) . $$

\begin{prop}\label{prop : not a lambda-model}
The $\lambda$-algebra $(\mathcal{P}(D), \ast, \llbracket - \rrbracket_-)$ is not a $\lambda$-model. 

In other words, there exist $ \rho \in {\mathcal{P}(D)}^\mathcal{V} $, $ x \in \mathcal{V} $ and two $\lambda$-terms $ t_1 $ and $ t_2 $ such that
$$ ((\forall d \in \mathcal{P}(D)) \llbracket t_1 \rrbracket_{\rho[x:=d]} = \llbracket t_2 \rrbracket_{\rho[x:=d]} \textrm{ and } \llbracket \lambda x. t_1 \rrbracket_\rho \not= \llbracket \lambda x. t_2 \rrbracket_\rho) . $$
\end{prop}

In particular, $ \llbracket t \rrbracket_\rho $ \emph{can not be} defined by induction on $ t $ (an interpretation by polynomials is nevertheless possible in such a way that the $\xi$-rule holds - see \cite{selinger}).

Before stating Proposition \ref{prop : not enough points}, we recall that any object $ A $ of any category $ \mathbb{K} $ with a terminal object is said \emph{to have enough points} if for any terminal object $ 1 $ of $ \mathbb{K} $ and for any $y, z \in \mathbb{K}(A, A) $, we have $ ((\forall x \in \mathbb{K}(1, A)) y \circ_\mathbb{K} x = z \circ_\mathbb{K} x \Rightarrow y=z) . $

Remark: it does not follow necessarily that the same holds for any $ y, z \in \mathbb{K}(A, B) $.

\begin{prop}\label{prop : not enough points}
Let $ A $ be a non-empty set. Then $ ( A ) $ does not have enough points in $ \Lambda(\mathfrak{M}) $.
\end{prop}

\begin{proof}
Let $ \alpha \in A$. Let $y$ and $z$ be the arrows $\mathcal{M}_{\textrm{fin}}(A) \rightarrow A$ of the category $\mathbf{Rel}$ defined by 
$y = \{ ([\alpha], \alpha) \}$ and $z = \{ ([\alpha, \alpha], \alpha) \}$.
Then $( y )$ and $( z )$ are two arrows $( A ) \rightarrow ( A )$ of the category $\Lambda(\mathfrak{M})$.

We recall that the terminal object in $\Lambda(\mathfrak{M})$ is the empty sequence $( )$. Now, for any arrow $x : ( ) \rightarrow ( A )$ of the category $\Lambda(\mathfrak{M})$, we have $( y ) \circ_{\Lambda(\mathfrak{M})} x = ( z ) \circ_{\Lambda(\mathfrak{M})} x$. 
\end{proof}

This proposition explains \emph{why} Proposition \ref{prop : not a lambda-model} holds. A more direct proof of Proposition \ref{prop : not a lambda-model} can be obtained by considering the two $\lambda$-terms $ t_1 = (y)x $ and $ t_2 = (z)x $ with $ \rho(y) = \{ ([\alpha], \alpha) \} $ and $ \rho(z) = \{ ([\alpha, \alpha], \alpha) \} $.

\subsection{Non-uniformity}\label{subsection:non uniformity}

Example \ref{example : non-uniformity} illustrates the non-uniformity of the semantics. It is based on the following idea. 

Consider the program

$\begin{array}{lll}
\lambda x. \textsf{if } x & \textsf{then } \mathbf{1} & \\
& \textsf{else if } x & \textsf{then } \mathbf{1} \\
& & \textsf{else } \mathbf{0}
\end{array}$

applied to a boolean. The second \textsf{then} is never read. A \emph{uniform} semantics would ignore it. It is not the case when the semantics is \emph{non-uniform}.

\begin{ex}\label{example : non-uniformity}
Set $ \mathbf{0} = \lambda x. \lambda y. y $ and $ \mathbf{1} = \lambda x. \lambda y. x $. Assume that $h$ is the inclusion from $\mathcal{M}_{\textrm{fin}}(D) \times D$ to $D$. 

Let $ \gamma \in D $; set $\delta = ([], ([\gamma], \gamma))$ and $ \beta = ([\gamma], ([], \gamma)) $. 
We have
\begin{itemize}
\item $([([], ([\delta], \delta))], ([\delta], \delta)) \in \llbracket (x) \mathbf{1} \rrbracket_x$; 
\item and $([([], ([\delta], \delta))], \delta) \in \llbracket (x) \mathbf{1} \mathbf{0} \rrbracket_x$.
\end{itemize}
Hence we have $\alpha_1 = ([([], ([\delta], \delta)), ([], ([\delta], \delta))], \delta) \in \llbracket \lambda x. (x) \mathbf{1}(x)\mathbf{1} \mathbf{0} \rrbracket . $ 

We have
\begin{itemize}
\item $ ([([], ([\beta], \beta))], ([\beta], \beta)) \in \llbracket (x) \mathbf{1} \rrbracket_x $;
\item and $ ([([\beta], ([], \beta))], \beta) \in \llbracket (x) \mathbf{1} \mathbf{0} \rrbracket_x . $
\end{itemize}
Hence we have $\alpha_2 = ([([], ([\beta], \beta)), ([\beta], ([], \beta))], \beta) \in \llbracket \lambda x. (x) \mathbf{1}(x)\mathbf{1} \mathbf{0} \rrbracket . $ 

In a uniform semantics (as in \cite{systemF}), the point $\alpha_1$ would appear in the semantics of this $\lambda$-term, but not the point $\alpha_2$, because $ [([], ([\beta], \beta)), ([\beta], ([], \beta))] $ corresponds to a chimerical argument: the argument is read twice and provides two contradictory values.
\end{ex}

%% file: execution_time_intersection.tex
\section{Non-idempotent intersection types}\label{section : System R}

From now on, $ D = \bigcup_{n \in \mathbb{N}} D_n $, where $ D_n $ is defined by induction on $ n $: $ D_0 $ is a non-empty set $ A $ that does not contain any pairs and $ D_{n+1} = A \cup (\mathcal{M}_{\textrm{fin}}(D_n) \times D_n)$. We have $D = A \dot{\cup} (\mathcal{M}_{\textrm{fin}}(D) \times D) $, where $ \dot{\cup}$ is the disjoint union; the injection $h$ from $\mathcal{M}_{\textrm{fin}}(D) \times D$ to $D$ will be the inclusion. Hence any element of $D$ can be written $a_1 \ldots a_m \alpha$, where $a_1, \ldots, a_m \in \mathcal{M}_{\textrm{fin}}(D)$, $\alpha \in D$ and $a_1 \ldots a_m \alpha$ is defined by induction on $m$:
\begin{itemize}
\item $a_1 \ldots a_0 \alpha = \alpha$;
\item $a_1 \ldots a_{m+1} \alpha = (a_1 \ldots a_m, (a_{m+1}, \alpha))$.
\end{itemize}

For any $\alpha \in D$, we denote by $\textsf{depth}(\alpha)$ the least integer $n$ such that $\alpha \in D_n$.

In the preceding section, we defined the semantics we consider (Definitions \ref{def : [t]} and 
\ref{Def : [t]_rho}). 
Now, we want to describe this semantics as a logical system: the elements of $ D $ are viewed as propositional formulas. More precisely, a comma separating a multiset of types and a type is understood as an arrow and a non-empty multiset is understood as the conjunction of its elements (their intersection). Note that this means we are considering a commutative (but not necessarily idempotent) intersection.

\subsection{System $R$}\label{subsection:System R}

A \emph{context} $ \Gamma $ is a function from $ \mathcal{V} $ to $\mathcal{M}_{\textrm{fin}}(D) $ such that $ \{ x \in \mathcal{V} \; / \; \Gamma(x) \not= [] \} $ is finite. If $ x_1, \ldots, x_m \in \mathcal{V} $ are distinct and $ a_1, \ldots, a_m \in \mathcal{M}_{\textrm{fin}}(D) $, then $ x_1 : a_1, \ldots, x_m : a_m $ denotes the context defined by 
$x \mapsto \left\lbrace \begin{array}{ll} a_j & \textrm{if $x = x_j$;}\\ \textrm{[]} & \textrm{else.}  \end{array} \right.$ 
We denote by $ \Phi $ the set of contexts. We define the following binary operation on $\Phi$: 
$$\begin{array}{rcl} \Phi \times \Phi & \rightarrow & \Phi \\ (\Gamma_1, \Gamma_2) & \mapsto & \Gamma_1 + \Gamma_2 : \begin{array}{rcl} \mathcal{V} & \rightarrow & \mathcal{M}_\textsf{fin}(D) \\ x & \mapsto & \Gamma_1(x)+\Gamma_2(x), \end{array} \end{array}$$ where the second $+$ denotes the sum of multisets given by term-by-term addition of multiplicities.
Note that this operation is associative and commutative. 
Typing rules concern judgements of the form $\Gamma \vdash_R t : \alpha$, where $\Gamma \in \Phi$, $t$ is a $\lambda$-term and $\alpha \in D$.

\begin{definition}\label{definition:System R}
The typing rules of System $ R $ are the following:

\begin{center}
\AxiomC{}
\UnaryInfC{$ x : [\alpha] \vdash_R x : \alpha $}
\DisplayProof
\end{center}
\begin{center}
\AxiomC{$ \Gamma, x : a \vdash_R v : \alpha $}
\UnaryInfC{$ \Gamma \vdash_R \lambda x. v : (a, \alpha) $}
\DisplayProof
\end{center}
\begin{center}
\AxiomC{$ \Gamma_0 \vdash_R v : ([\alpha_1, \ldots, \alpha_n], \alpha) 
$}
\AxiomC{$ \Gamma_1 \vdash_R u : \alpha_1, \ldots, \Gamma_n \vdash_R u :
  \alpha_n $}
\RightLabel{$ n \in \mathbb{N} $}
\BinaryInfC{$ \Gamma_0 + \Gamma_1 + \ldots + \Gamma_n \vdash_R (v)u :
  \alpha $}
\DisplayProof
\end{center}
\end{definition}

The typing rule of the application has $ n+1 $ premisses. In particular, in the case where $ n = 0 $, we obtain the following rule:  \AxiomC{$ \Gamma_0 \vdash_R v : ([],\alpha) $} \UnaryInfC{$ \Gamma_0 \vdash_R (v)u : \alpha $} \DisplayProof for any $\lambda$-term $ u $. So, the empty multiset plays the role of the universal type $ \Omega $.

The intersection we consider is \emph{not} idempotent in the following sense: if a closed $\lambda$-term $ t $ has the type $ a_1 \ldots a_m \alpha $ and, for $ 1 \leq j \leq m $, $ \textsf{Supp}(a'_j) = \textsf{Supp}(a_j) $, it does not follow necessarily that $ t $ has the type $ a'_1 \ldots a'_m \alpha $. 
For instance, the $\lambda$-term $\lambda z. \lambda x. (z)x$ has types $( [ ([\alpha], \alpha) ], ([\alpha], \alpha) )$ and $( [ ([\alpha, \alpha], \alpha) ], ([\alpha, \alpha], \alpha) )$ but not the type $( [ ([\alpha], \alpha) ], ([\alpha, \alpha], \alpha) )$. 
On the contrary, the system presented in \cite{Ronchi} and the System $ \mathcal{D} $ presented in \cite{typesetmodeles} consider an idempotent intersection. System $ \lambda $ of \cite{Kfoury} and System $ \mathbb{I} $ of \cite{NeeMairson} consider a non-idempotent intersection, but the treatment of weakening is not the same.

Interestingly, System $ R $ can be seen as a reformulation of the system of \cite{CDV}. More precisely, types of System $ R $ correspond to their normalized types. 

\subsection{Relating types and semantics}\label{subsection : relating types and semantics}

We prove in this subsection that the semantics of a closed $\lambda$-term as defined in Subsection \ref{subsection:untyped} is the set of its types in System $R$. The following assertions relate more precisely types and semantics of any $\lambda$-term.

\begin{theorem}\label{Theo : semantique = types}
For any $\lambda$-term $ t $ such that $ FV(t) \subseteq \{ x_1, \ldots, x_m \} $, we have
$$ \llbracket t \rrbracket_{x_1, \ldots, x_m} = \{ a_1 \ldots a_m \alpha \in (\prod_{j=1}^m \mathcal{M}_{\textrm{fin}}(D)) \times D \; / \; x_1 : a_1, \ldots, x_m : a_m \vdash_R 
t : \alpha \}  . $$
\end{theorem}

\begin{proof}
By induction on $ t $.
\end{proof}

\begin{corollaire}\label{corollary : subject conversion}
For any $\lambda$-terms $ t $ and $ t' $ such that $ t =_\beta t' $, if $ \Gamma \vdash_R t : \alpha $, then we have $ \Gamma \vdash_R t' : \alpha $.
\end{corollaire}

\begin{proof}
By our Proposition \ref{prop : CCC} and Lemma \ref{lemma : retraction}, and Proposition 5.5.5 of \cite{Barendregt}, the following property holds: for any $\lambda$-terms $ t $ and $ t' $ such that $ t =_\beta t' $ and such that $ FV(t) \subseteq \{ x_1, \ldots, x_m \} $, we have $\llbracket t \rrbracket_{x_1, \ldots, x_m} = \llbracket t' \rrbracket_{x_1, \ldots, x_m}$.
\end{proof}

\begin{theorem}\label{Theo : systeme -> modele}
For any $\lambda$-term $ t $ and for any $ \Gamma \in \Phi $, we have 
\begin{eqnarray*}
& & \{ \alpha \in D \: / \: \Gamma \vdash_R t : \alpha \} \\
& \subseteq & \{ \alpha \in D \: / \: (\forall \rho \in {\mathcal{P}(D)}^\mathcal{V}) ((\forall x \in \mathcal{V}) \Gamma(x) \in \mathcal{M}_{\textrm{fin}}(\rho(x)) \Rightarrow \alpha \in \llbracket t \rrbracket_\rho ) \} .
\end{eqnarray*}
\end{theorem}

\begin{proof}
Apply Theorem \ref{Theo : semantique = types}.
\end{proof}

\begin{rem}
The reverse inclusion is not true.
\end{rem}

\begin{theorem}
For any $\lambda$-term $ t $ and for any $ \rho \in {\mathcal{P}(D)}^\mathcal{V} $, we have
$$\llbracket t \rrbracket_\rho = \{ \alpha \in D \: / \: (\exists \Gamma \in \Phi) ((\forall x \in \mathcal{V}) \Gamma(x) \in \mathcal{M}_{\textrm{fin}}(\rho(x)) \textrm{ and } \Gamma \vdash_R t : \alpha) \} \enspace . $$
\end{theorem}

\begin{proof}
Apply Theorems \ref{Theo : semantique = types} and \ref{Theo : systeme -> modele}.
\end{proof}

There is another way to compute the interpretation of $\lambda$-terms in this semantics. Indeed, it is well-known that we can translate $\lambda$-terms into linear logic nets labelled with the types $I$, $O$, $?I$ and $!O$ (as in \cite{regnier}): this translation is defined by induction on the $\lambda$-terms. Now, we can do \emph{experiments} (in the sense of \cite{LinearLogic}, that introduced this notion in the framework of coherent semantics for working with proof-nets directly, without sequentializing)  to compute the semantics of the net in the multiset based relational model: all the translations corresponding to the encoding $ A \Rightarrow B \equiv ? A^\perp \wp B $ have the same semantics. And this semantics is the same as the semantics defined here.

For a survey of translations of $\lambda$-terms in proof nets, see \cite{lambdatoproofs}.

\subsection{An equivalence relation on derivations}\label{subsection : resource terms}

Definition \ref{definition : equivalence relation} introduces an equivalence relation on the set of derivations of a given $\lambda$-term. This relation, as well as the notion of substitution defined immediately after, will play a role in Subsection \ref{subsection:exact_number}.

\begin{definition}
For any $\lambda$-term $t$, for any $(\Gamma, \alpha) \in \Phi \times D$, we denote by $\Delta(t, (\Gamma, \alpha))$ the set of derivations of $\Gamma \vdash_R t : \alpha$.

For any $\lambda$-term $t$, we set $\Delta(t) = \bigcup_{(\Gamma, \alpha) \in \Phi \times D} \Delta(t, (\Gamma, \alpha))$.

For any closed $\lambda$-term $ t $, for any $ \alpha \in D $, we denote by $ \Delta(t, \alpha) $ the set of derivations of $ \vdash_R t : \alpha $.

For any closed $\lambda$-term $t$, for any integer $n$, for any $a \in \mathcal{M}_n(D)$, we set 
$$\Delta(t, a) = \bigcup_{(\alpha_1, \ldots, \alpha_n) \in \mathfrak{S}(a)} \{ (\Pi_1, \ldots, \Pi_n) \in \Delta(t)^n \: / \: (\forall i \in \{ 1, \ldots, n \}) \: \Pi_i \in \Delta(t, \alpha_i) \} \enspace .$$

We set $\Delta = \bigcup_{t \in \Lambda} \Delta(t)$.
\end{definition}

\begin{definition}\label{definition : equivalence relation}
Let $t$ be a $\lambda$-term. For any $\Pi, \Pi' \in \Delta(t)$, we define, by induction on $\Pi$, when $\Pi \sim \Pi'$ holds:
\begin{itemize}
\item if $ \Pi $ is only a leaf, then $ \Pi \sim \Pi' $ if, and only if, $ \Pi' $ is a leaf too;
\item if $ \Pi = \AxiomC{$ \begin{array}{c} \Pi_0 \\ \Gamma, x : a \vdash_R v : \alpha \end{array}$}
\UnaryInfC{$ \Gamma \vdash_R \lambda x. v : (a, \alpha) $}
\DisplayProof,$ then $\Pi \sim \Pi'$ if, and only if, there exists $\Pi_0' \sim \Pi_0$ such that 
$\Pi' = \AxiomC{$ \begin{array}{c} \Pi_0' \\ \Gamma', x : a' \vdash_R v : \alpha' \end{array}$}
\UnaryInfC{$ \Gamma' \vdash_R \lambda x. v : (a', \alpha') $}
\DisplayProof$;
\item if $$ \Pi = 
 \AxiomC{$ \begin{array}{c} \Pi_0 \\ \Gamma_0 \vdash_R v : ([\alpha_1, \ldots, \alpha_n], \alpha) \end{array} 
$}
\AxiomC{$ \begin{array}{ccc} \Pi_1 & \ldots & \Pi_n \\ \Gamma_1 \vdash_R u : \alpha_1 & \ldots & \Gamma_n \vdash u :
  \alpha_n \end{array} $}
\BinaryInfC{$ \Gamma_0 + \Gamma_1 + \ldots + \Gamma_n \vdash_R (v)u :
  \alpha $}
\DisplayProof \enspace ,$$ 
then $ \Pi \sim \Pi' $ if, and only if, there exist $\Pi_0 \sim \Pi_0'$, $\sigma \in \mathfrak{S}_n$, $\Pi_1 \sim \Pi'_{\sigma(1)}, \ldots, \Pi_n \sim \Pi'_{\sigma(n)}$ such that
$$ \Pi' = \AxiomC{$ \begin{array}{c} \Pi_0' \\ \Gamma_0' \vdash_R v : ([\alpha_1', \ldots, \alpha_n'], \alpha') \end{array}
$}
\AxiomC{$ \begin{array}{ccc} \Pi_1' & \ldots & \Pi_n' \\ \Gamma_1' \vdash_R u : \alpha_1' & \ldots & \Gamma_n' \vdash u :
  \alpha_n' \end{array} $}
\BinaryInfC{$ \Gamma_0' + \Gamma_1' + \ldots + \Gamma_n' \vdash_R (v)u :
  \alpha' \enspace , $}
\DisplayProof  \enspace .$$
\end{itemize}
\end{definition}

An equivalence class of derivations of a $\lambda$-term $ t $ in System $ R $ can be seen as a \emph{simple resource term of the shape of} $ t $ that does not reduce to $ 0 $. \emph{Resource $\lambda$-calculus} is defined in \cite{thomas} and is similar to resource oriented versions of the $\lambda$-calculus previously introduced and studied in \cite{BCL} and \cite{Kfoury}. For a full exposition of a precise relation between this equivalence relation and simple resource terms, see \cite{these}.

\begin{definition}\label{definition : substitution}
A substitution $\sigma$ is a function from $ D $ to $ D $ such that 
$$ \textrm{for any $ \alpha, \alpha_1, \ldots, \alpha_n \in D $, } \sigma([\alpha_1, \ldots, \alpha_n], \alpha) = ([\sigma(\alpha_1), \ldots, \sigma(\alpha_n)], \sigma(\alpha)) \enspace . $$ 

We denote by $\mathcal{S}$ the set of substitutions.

For any $\sigma \in \mathcal{S}$, we denote by $ \overline{\sigma} $ the function from $ \mathcal{M}_{\textrm{fin}}(D) $ to $ \mathcal{M}_{\textrm{fin}}(D) $ defined by $ \overline{\sigma}([\alpha_1, \ldots, \alpha_n]) = [\sigma(\alpha_1), \ldots, \sigma(\alpha_n)] $.
\end{definition}

\begin{prop}\label{proposition : substitution and equivalence class}
Let $ \Pi $ be a derivation of $ \Gamma \vdash_R t : \alpha $ and let $ \sigma $ be a substitution. Then there exists a derivation $ \Pi' $ of $ \overline{\sigma} \circ \Gamma \vdash_R t : \sigma(\alpha) $ such that $ \Pi \sim \Pi' $.
\end{prop}

\begin{proof}
By induction on $ t $.
\end{proof}

%% file: execution_time_qualitative.tex
\section{Qualitative results}\label{section:qualitative}

In this section, inspired by \cite{typesetmodeles}, we prove Theorem \ref{th normalisable de tete}, which formulates \emph{qualitative} relations between assignable types and normalization properties: it characterizes the (head) normalizable $\lambda$-terms by semantic means. 
We also answer to the following question: if $v$ and $u$ are two closed normal $\lambda$-terms, is it the case that $(v)u$ is (head) normalizable? The answer is given only referring to $\llbracket v \rrbracket$ and $\llbracket u \rrbracket$ in Corollary \ref{cor:qualitative}. Quantitative versions of this last result will be proved in Section \ref{section:quantitative}.

\begin{definition}
For any $n \in \mathbb{N}$, we define, by induction on $n$, $D_n^\textsf{ex}$ and $\overline{D_n^\textsf{ex}}$:
\begin{itemize}
\item $D_0^\textsf{ex} = \overline{D_0^\textsf{ex}} = A$;
\item $D_{n+1}^\textsf{ex} = A \cup (\mathcal{M}_\textrm{fin}(\overline{D_n^\textsf{ex}}) \times D_n^\textsf{ex})$ and $\overline{D_{n+1}^\textsf{ex}} = A \cup ((\mathcal{M}_\textrm{fin}(D_n^\textsf{ex}) \setminus \{ [] \}) \times \overline{D_n^\textsf{ex}})$.
\end{itemize}
We set 
\begin{itemize}
\item $D^\textsf{ex} = \bigcup_{n \in \mathbb{N}} D_n^\textsf{ex}$;
\item $\overline{D^\textsf{ex}} = \bigcup_{n \in \mathbb{N}} \overline{D_n^\textsf{ex}}$;
\item and $\Phi^\textsf{ex} = \{ \Gamma \in \Phi \: / \: (\forall x \in \mathcal{V}) \Gamma(x) \in \mathcal{M}_\textrm{fin}(\overline{D^\textsf{ex}}) \}  .$
\end{itemize}
\end{definition}

Note that $D^\textsf{ex}$ is the set of the $\alpha \in D$ such that $[]$ has no positive occurrences in $\alpha$.

\renewcommand{\labelenumi}{(\roman{enumi})}

\begin{prop}\label{prop : head => typable}
\begin{enumerate}
\item Every head-normalizable $\lambda$-term is typable in System $ R $.
\item For any normalizable $\lambda$-term $ t $, there exists $(\Gamma, \alpha) \in \Phi^\textsf{ex} \times D^\textsf{ex}$ 
such that $ \Gamma \vdash_R t : \alpha $. 
\end{enumerate}
\end{prop}

\begin{proof}
\begin{enumerate}
\item Let $ t $ be a head-normalizable $\lambda$-term. There exists a $\lambda$-term of the shape $(\lambda x_1. \ldots \lambda x_k. t)v_1\ldots v_n$ such that $(\lambda x_1. \ldots \lambda x_k. t)v_1\ldots v_n =_\beta x $. Now, $ x $ is typable. Therefore, by Corollary~\ref{corollary : subject conversion}, the $ \lambda $-term $ (\lambda x_1. \ldots \lambda x_k. t)v_1 \ldots v_n $ is typable. Hence $ \lambda x_1. \ldots \lambda x_k. t $ is typable.
\item We prove, by induction on $ t $, that for any normal $\lambda$-term $ t $, the following properties hold:
\begin{itemize}
\item there exists $(\Gamma, \alpha) \in \Phi^\textsf{ex} \times D^\textsf{ex}$ such that $ \Gamma \vdash_R t : \alpha $;
\item if, moreover, $ t $ does not begin with $ \lambda $, then, for any $ \alpha \in D^\textsf{ex} $, there exists $ \Gamma \in \Phi^\textsf{ex}$ such that $ \Gamma \vdash_R t : \alpha $. 
\end{itemize}
Next, just apply Corollary \ref{corollary : subject conversion}.
\end{enumerate}
\end{proof}

If $ \mathcal{X}_1 $ and $ \mathcal{X}_2 $ are two sets of $\lambda$-terms, then $ \mathcal{X}_1 \rightarrow \mathcal{X}_2 $ denotes the set of $ \lambda $-terms $ v $ such that for any $ u \in \mathcal{X}_1 $, we have $(v)u \in \mathcal{X}_2 $. 
A set $ \mathcal{X} $ of $\lambda$-terms is said to be \emph{saturated} if for any $\lambda$-terms $ t_1, \ldots, t_n, u $ and for any $ x \in \mathcal{V} $, we have 
$$ ((u[t/x])t_1 \ldots t_n \in \mathcal{X} \Rightarrow (\lambda x. u)tt_1 \ldots t_n \in \mathcal{X}) .$$ 
An interpretation is a map from $ A $ to the set of saturated sets. For any interpretation $ \mathcal{I} $ and for any $ \delta \in D \cup \mathcal{M}_{\textrm{fin}}(D) $, we define, by induction on $ \delta $, a saturated set $ \vert \delta \vert_\mathcal{I} $:
\begin{itemize}
\item if $ \delta \in A $, then $ \vert \delta \vert_\mathcal{I} = \mathcal{I}(\delta) $;
\item if $ \delta = [] $, then $ \vert \delta \vert_\mathcal{I} $ is the set of all $\lambda$-terms;
\item if $ \delta = [\alpha_1, \ldots, \alpha_{n+1}] $, then $ \vert \delta \vert_\mathcal{I} = \bigcap_{i=1}^{n+1} \vert \alpha_i \vert_\mathcal{I} $. 
\item if $ \delta = (a, \alpha) $, then $ \vert \delta \vert_\mathcal{I} = \vert a \vert_\mathcal{I} \rightarrow \vert \alpha \vert_\mathcal{I} $.
\end{itemize}

\begin{lemme}\label{adequation}
Let $ \mathcal{I} $ be an interpretation and let $ u $ be a $\lambda$-term such that $ x_1 : a_1, \ldots, x_k : a_k \vdash_R u :\alpha $. If $ t_1 \in \vert a_1 \vert_\mathcal{I}, \ldots $, $ t_k \in \vert a_k \vert_\mathcal{I} $, then $ u[t_1/x_1, \ldots, t_k/x_k] \in \vert \alpha \vert_\mathcal{I} $.
\end{lemme}

\begin{proof}
By induction on $ u $.
\end{proof}

\begin{lemme}\label{head-normalizable : lemme 2}
\begin{enumerate}
\item Let $ \mathcal{N} $ be the set of head-normalizable terms. For any $ \gamma \in A $, we set $ \mathcal{I}(\gamma) = \mathcal{N} $. Then, for any $ \alpha \in D $, we have $\mathcal{V} \subseteq \vert \alpha \vert_\mathcal{I} \subseteq \mathcal{N} $.
\item Let $ \mathcal{N} $ be the set of normalizable terms.  For any $ \gamma \in A $, we set \mbox{$ \mathcal{I}(\gamma) = \mathcal{N} $}. For any $ \alpha \in \overline{D^\textsf{ex}} $ (resp. $\alpha \in D^\textsf{ex}$), we have \mbox{$ \mathcal{V} \subseteq \vert \alpha \vert_\mathcal{I} $} (resp. $ \vert \alpha \vert_\mathcal{I} \subseteq \mathcal{N} $).
\end{enumerate}
\end{lemme}

\begin{proof}
\begin{enumerate}
\item Set $\mathcal{N}_0 = \{ (x)t_1 \ldots t_n \: / \: x \in \mathcal{V} \textrm{ and } t_1, \ldots, t_n \in \Lambda \}$. We prove, by induction on $\alpha$, that we have $\mathcal{N}_0 \subseteq \vert \alpha \vert_\mathcal{I} \subseteq \mathcal{N}$.

If $ \alpha = (b, \beta) $, then, by induction hypothesis, we have $ \mathcal{N}_0 \subseteq \vert \beta \vert_\mathcal{I} \subseteq \mathcal{N} $ and $ \mathcal{N}_0 \subseteq \vert b \vert_\mathcal{I} $. Hence we have $ \mathcal{N}_0 \subseteq \Lambda \rightarrow \mathcal{N}_0 \subseteq \vert \alpha \vert_\mathcal{I} $ and $ \vert \alpha \vert_\mathcal{I} \subseteq \mathcal{N}_0 \rightarrow \mathcal{N} \subseteq \mathcal{N} $.
\item Set $ \mathcal{N}_0 = \{ (x)t_1 \ldots t_n \: / \: x \in \mathcal{V} \textrm{ and } t_1, \ldots, t_n \in \mathcal{N} \} $. We prove, by induction on $ \alpha $, that
\begin{itemize}
\item if $\alpha \in \overline{D^\textsf{ex}}$, then we have $\mathcal{N}_0 \subseteq \vert \alpha \vert_\mathcal{I}$;
\item if $\alpha \in D^\textsf{ex}$, then we have $\vert \alpha \vert_\mathcal{I} \subseteq \mathcal{N}$.
\end{itemize}
Suppose $ \alpha = (b, \beta) \in \mathcal{M}_{\textrm{fin}}(D) \times D $.
\begin{itemize}
\item If $\alpha \in \overline{D^\textsf{ex}}$, then $b \in \mathcal{M}_\textrm{fin}(D^\textsf{ex})$ and $\beta \in \overline{D^\textsf{ex}}$. By induction hypothesis, we have $ \vert b \vert_\mathcal{I} \subseteq \mathcal{N} $ and $ \mathcal{N}_0 \subseteq \vert \beta \vert_\mathcal{I} $. Hence $ \mathcal{N}_0 \subseteq \mathcal{N} \rightarrow \mathcal{N}_0 \subseteq \vert b \vert_\mathcal{I} \rightarrow \vert \beta \vert_\mathcal{I} = \vert \alpha \vert_\mathcal{I} $.
\item If $\alpha \in D^\textsf{ex}$, then $b \in \mathcal{M}_\textrm{fin}(\overline{D^\textsf{ex}})$ and $\beta \in D^\textsf{ex}$. By induction hypothesis, we have $ \mathcal{N}_0 \subseteq \vert b \vert_\mathcal{I} $ and $ \vert \beta \vert_\mathcal{I} \subseteq \mathcal{N} $. Hence $ \vert \alpha \vert_\mathcal{I} = \vert b \vert_\mathcal{I} \rightarrow \vert \beta \vert_\mathcal{I} \subseteq \mathcal{N}_0 \rightarrow \mathcal{N} \subseteq \mathcal{N} $ (this last inclusion follows from the fact that for any $\lambda$-term $t$, for any variable $x$ that is not free in $t$, if $(t)x$ is normalizable, then $t$ is normalizable, fact that can be proved by induction on the number of left-reductions of $(t)x$).
\end{itemize}
\end{enumerate}
\end{proof}

\begin{prop}\label{prop : typable => head}
\begin{enumerate}
\item Every typable $\lambda$-term in System $ R $ is head-normalizable.
\item Let $t \in \Lambda$, $\alpha \in D^\textsf{ex}$ and $\Gamma \in \Phi^\textsf{ex}$ such that $ \Gamma \vdash_R t : \alpha $. Then $ t $ is normalizable. 
\end{enumerate}
\end{prop}

\begin{proof}
\begin{enumerate}
\item Let $\Gamma$ be the context $x_1 : a_1, \ldots, x_k : a_k$. For any $ \gamma \in A $, we set $ \mathcal{I}(\gamma) = \mathcal{N} $, where $ \mathcal{N} $ is the set of head-normalizable terms. By Lemma \ref{head-normalizable : lemme 2} (i), we have $x_1 \in \vert a_1 \vert_\mathcal{I} $, $\ldots$, $x_k \in \vert a_k \vert_\mathcal{I}$. Hence, by Lemma \ref{adequation}, we have $ t = t[x_1/x_1, \ldots, x_k/x_k] \in \vert \alpha \vert_\mathcal{I} $. Using again Lemma \ref{head-normalizable : lemme 2} (i), we obtain $ \vert \alpha \vert_\mathcal{I} \subseteq \mathcal{N} $.
\item Let $\Gamma$ be the context $ x_1 : a_1, \ldots, x_k : a_k $. For any $ \gamma \in A $, we set $ \mathcal{I}(\gamma) = \mathcal{N} $, where $ \mathcal{N} $ is the set of normalizable terms. By Lemma \ref{head-normalizable : lemme 2} (ii), we have $x_1 \in \vert a_1 \vert_\mathcal{I}$, $\ldots$, $x_k \in \vert a_k \vert_\mathcal{I}$. Hence, by Lemma \ref{adequation}, we have $ t = t[x_1/x_1, \ldots, x_k/x_k] \in \vert \alpha \vert_\mathcal{I} $. Using again Lemma \ref{head-normalizable : lemme 2} (ii), we obtain $ \vert \alpha \vert_\mathcal{I} \subseteq \mathcal{N} $.
\end{enumerate}
\end{proof}

\begin{theorem}\label{th normalisable de tete}
\begin{enumerate}
\item For any $ t \in \Lambda $, $ t $ is head-normalizable if, and only if, $ t $ is typable in System $ R $.
\item For any $ t \in \Lambda $, $ t $ is normalizable if, and only if, there exist $ (\Gamma, \alpha) \in \Gamma^\textsf{ex} \times D^\textsf{ex}$ such that $ \Gamma \vdash_R t : \alpha $.
\end{enumerate}
\end{theorem}

\begin{proof}
\begin{enumerate}
\item Apply Proposition \ref{prop : head => typable} (i) and Proposition \ref{prop : typable => head} (i).
\item Apply Proposition \ref{prop : head => typable} (ii) and Proposition \ref{prop : typable => head} (ii).
\end{enumerate}
\end{proof}

This theorem is not surprising: although System $ R $ is not considered in \cite{Honsell}, it is quite obvious that its typing power is the same as that of the systems containing $ \Omega $ considered in this paper. We can note here a difference with Systems $ \lambda $ and $ \mathbb{I} $ already mentioned: in those systems, only strongly normalizable terms are typable. Of course, such systems characterizing the strongly normalizable terms, cannot be in correspondence with a denotational semantics of $\lambda$-calculus.

\begin{corollaire}\label{cor:qualitative}
Let $v$ and $u$ two closed normal terms.
\begin{enumerate}
\item The $\lambda$-term $(v)u$ is head-normalizable if, and only if, there exist $a \in \mathcal{M}_{\textrm{fin}}(\llbracket u \rrbracket)$ and $\alpha \in D$ such that $(a, \alpha) \in \llbracket v \rrbracket$.
\item The $\lambda$-term $(v)u$ is normalizable if, and only if, there exist $a \in \mathcal{M}_{\textrm{fin}}(\llbracket u \rrbracket)$ and $\alpha \in D^\textsf{ex}$ such that $(a, \alpha) \in \llbracket v \rrbracket$.
\end{enumerate}
\end{corollaire}

%% file: execution_time_quantitative.tex
\section{Quantitative results}\label{section:quantitative}

We now turn our attention to the quantitative aspects of reduction. The aim is to give a purely semantic account of execution time. Of course, if $t'$ is the normal form of $t$, we know that $\llbracket t \rrbracket = \llbracket t' \rrbracket$, so that from $\llbracket t \rrbracket$ it is clearly impossible to determine the number of reduction steps from $t$ to $t'$. Nevertheless, if $v$ and $u$ are two normal $\lambda$-terms, we can wonder what is the number of steps leading from $(v)u$ to its (principal head) normal form. We prove in this section that we can answer the question by only referring to $\llbracket v \rrbracket$ and $\llbracket u \rrbracket$ (Theorem \ref{theorem semantics}).

\subsection{Type Derivations for States}

We now extend the type derivations for $\lambda$-terms to type derivations for closures (Definition \ref{definition : derivation for closures}) and for states (Definition \ref{definition : derivation for states}). We will define also the size $\vert \Pi \vert$ of such derivations $\Pi$. Naturally, the size $\vert \Pi \vert$ of a derivation $\Pi$ of System R is quite simply its size as a tree, i.e. the number of its nodes; moreover, for any $n \in \mathbb{N}$, for any $(\Pi_1, \ldots, \Pi_n) \in \Delta^n$, we set $\vert (\Pi_1, \ldots, \Pi_n) \vert = \sum_{i=1}^n \vert \Pi_i \vert$.

\begin{definition}\label{definition : derivation for closures} 
For any closure $c = (t, e)$, for any $(\Gamma, \alpha) \in \Phi \times D$ (respectively $(\Gamma, a) \in \Phi \times \mathcal{M}_{\textrm{fin}}(D)$), we define, by induction on $\textsf{d}(e)$, what is a derivation $\Pi$ of $\Gamma \vdash c : \alpha$ (respectively $\Gamma \vdash c : a$) and what is $\vert \Pi \vert$ for such a derivation:
\begin{itemize}
\item A derivation of $\Gamma \vdash (t, \bigcup_{j=1}^m \{ (x_j, c_j) \} ) : \alpha$ is a pair $(\Pi_0, \bigcup_{j=1}^m \{ (x_j, \Pi_j) \} ) $, where
\begin{itemize}
\item $\Pi_0$ is a derivation of $\Gamma_0, \: x_1 : a_1, \ldots, x_m : a_m \vdash_R t:\alpha$;
\item for any $ j \in \{1, \ldots, m \}$, $\Pi_j$ is a derivation of $ \Gamma_j \vdash c_j : a_j $;
\item and $\Gamma = \sum_{j=0}^m \Gamma_j$.
\end{itemize}
If $ \Pi = (\Pi_0, \bigcup_{j=1}^m \{ (x_j, \Pi_j) \}) $ is a derivation of $ \Gamma \vdash c : \alpha $, then we set $ \vert \Pi \vert = \sum_{j=0}^m \vert \Pi_j \vert $. 
\item For any integer $ p $, a derivation of $ \Gamma \vdash c : [\alpha_1, \ldots, \alpha_p] $ is a $ p $-tuple $ (\Pi^1, \ldots, \Pi^p) $ such that there exists $ (\Gamma^1, \ldots, \; \Gamma^p) \in \Phi^p $ and  
\begin{itemize}
\item for $ 1 \leq i \leq p $, $ \Pi^i $ is a derivation of $ \Gamma^i \vdash c : \alpha_i $;
\item and $ \Gamma = \sum_{i=1}^p \Gamma^i $.
\end{itemize}
If $ \Pi = (\Pi^1, \ldots, \Pi^p) $ is a derivation of $ \Gamma \vdash c : a $, then we set $ \vert \Pi \vert = \sum_{i=1}^p \vert \Pi^i \vert $.
\end{itemize}
\end{definition}

Definition \ref{definition : derivation for closures} is not so easy to use directly. This is why we introduce Lemmas \ref{KAM = D cas application} and \ref{lemme : derivation cloture : cas abstraction}, that will be useful for proving Propositions \ref{Prop KAM = D} and \ref{prop:normal_principal}.

\begin{lemme}\label{KAM = D cas application}
Let $ ((v)u, e) \in \mathcal{C} $. For any $ b \in \mathcal{M}_{\textrm{fin}}(D) $, $ \Gamma', $ $ \Gamma'' \in \Phi $, if $ \Pi' $ is a derivation of $ \Gamma' \vdash (v, e) : (b, \alpha) $ and $ \Pi'' $ is a derivation of $ \Gamma'' \vdash (u, e) : b $, then there exists a derivation $ \Pi $ of $
\Gamma'+ \Gamma'' \vdash ((v)u, e) : \alpha $ such that $ \vert \Pi \vert = \vert \Pi' \vert + \vert \Pi'' \vert +1 $.
\end{lemme}

\begin{proof}
Set $ e = \bigcup_{j=1}^m \{ (x_j, c_j) \} $ and $ \Pi' = (\Pi'_0, \bigcup_{j=1}^m \{ (x_j, \Pi'_j) \}) $, where
\begin{itemize}
\item $ \Pi'_0 $ is a derivation of $ \Gamma'_0, x_1 : a'_1, \ldots, x_m : a'_m \vdash_R v : (b, \alpha) $;
\item for $ 1 \leq j \leq m $, $ \Pi'_j $ is a derivation of $ \Gamma'_j \vdash c_j : a'_j $;
\item and $ \Gamma' = \sum_{j=0}^m \Gamma'_j $. 
\end{itemize}
Set $ b = [\beta_1, \ldots, \beta_p ] $ and $ \Pi'' = ({\Pi''}^1, \ldots, {\Pi''}^p) $ where, for $ 1 \leq k \leq p $, 
$ {\Pi''}^k $ is a derivation $({\Pi''}_0^k, \bigcup_{j=1}^m \{ (x_j, {\Pi''}_j^k) \} )$
of $ {\Gamma''}^k \vdash (u, e) : \beta_k $ with $ \Gamma'' = \sum_{k=1}^p {\Gamma''}^k $. For $ k \in \{ 1, \ldots, p \} $, 
\begin{itemize}
\item $ {\Pi''}_0^k $ is a derivation of $ {\Gamma''}_0^k, x_1 : {a''}_1^k, \ldots, x_m : {a''}_m^k \vdash_R u : \beta_k $;
\item for $ 1 \leq j \leq m $, $ {\Pi''}_j^k $ is a derivation of $ {\Gamma''}_j^k \vdash c_j : {a''}_j^k $;
\item and $ {\Gamma''}^k = \sum_{j=0}^m {\Gamma''}_j^k $. 
\end{itemize}
For $ j \in \{ 0, \ldots, m \} $, we set $ \Gamma_j = \Gamma'_j + \sum_{k=1}^p {\Gamma''}_j^k $ and $ a_j = a'_j + \sum_{k=1}^p {a''}_j^k $. There exists a derivation $ \Pi_0 $ of $ \Gamma_0, x_1 : a_1, \ldots, x_m : a_m \vdash_R (v)u : \alpha $ with $ \vert \Pi_0 \vert = \vert \Pi'_0 \vert + \sum_{k=1}^p \vert {\Pi''}_0^k \vert + 1$. Moreover, for $j \in \{ 1, \ldots, m \}$, $ \Pi_j = \Pi'_j \ast {\Pi''}_j^1 \ast \ldots \ast {\Pi''}_j^p $, where $ \ast $ is the concatenation of finite sequences, is a derivation of $ \Gamma_j \vdash c_j : a_j $. We have
\begin{eqnarray*}
\sum_{j=0}^m \Gamma_j & = & \sum_{j=0}^m (\Gamma'_j + \sum_{k=1}^p {\Gamma''}_j^k) \allowdisplaybreaks \\
& = & \sum_{j=0}^m \Gamma'_j + \sum_{j=0}^m \sum_{k=1}^p {\Gamma''}_j^k \allowdisplaybreaks \\ 
& = & \Gamma' + \Gamma'' .
\end{eqnarray*}
Hence $ \Pi = (\Pi_0, \bigcup_{j=1}^m \{ (x_j, \Pi_j) \}) $ is a derivation of $
\Gamma'+ \Gamma'' \vdash ((v)u, e) : \alpha $. We have 
\begin{eqnarray*}
\vert \Pi \vert & = & \sum_{j=0}^m \vert \Pi_j \vert \allowdisplaybreaks \\
& = & \vert \Pi'_0 \vert + \sum_{k=1}^p \vert {\Pi''}_0^k \vert + 1 + \sum_{j=1}^m (\vert \Pi'_j \vert + \sum_{k=1}^p \vert {\Pi''}_j^k \vert) \allowdisplaybreaks \\
& = & \vert \Pi'\vert + \sum_{k=1}^p \vert {\Pi''}^k \vert +1 \allowdisplaybreaks \\
& = & \vert \Pi'\vert + \vert \Pi'' \vert + 1 . 
\end{eqnarray*}
\end{proof}

\begin{lemme}\label{lemme : derivation cloture : cas abstraction}
For any closure $ (u, e) $, for any derivation $ \Pi' $ of $ \Gamma, x : b \vdash (u, e) : \beta $, there exists a derivation $ \Pi $ of $ \Gamma \vdash (\lambda x. u, e) : (b, \beta) $ such that $\vert \Pi \vert = \vert \Pi' \vert + 1 $.
\end{lemme}

\begin{proof}
We set $ e = \bigcup_{j=1}^m \{ (x_j, c_j) \} $ and $ \Pi' = (\Pi'_0, \bigcup_{j=1}^m \{ (x_j, \Pi'_j) \} ) $. We know that $ \Pi'_0 $ is a derivation of $ \Gamma, x : b, x_1 : a_1, \ldots, x_m : a_m \vdash_R u : \beta $, hence there exists a derivation $ \Pi_0 $ of $ \Gamma, x_1 : a_1, \ldots, x_m : a_m \vdash_R \lambda x. u : (b, \beta) $. We set $ \Pi = (\Pi_0, \bigcup_{j=1}^m \{ (x_j, \Pi'_j) \} ) $: it is a derivation of $ \Gamma \vdash (\lambda x. u, e) : (b, \beta) $ and we have
\begin{eqnarray*}
\vert \Pi \vert & = & \vert \Pi_0 \vert + \sum_{j=1}^m \vert \Pi'_j \vert \allowdisplaybreaks \\
& = & \vert \Pi'_0 \vert + 1 + \sum_{j=1}^m \vert \Pi'_j \vert \allowdisplaybreaks \\
& = & \vert \Pi' \vert + 1 . 
\end{eqnarray*}
\end{proof}

\begin{definition}\label{definition : derivation for states}
Let $ s = (c_0, \ldots, c_q) $ be a state. A finite sequence $ (\Pi_0, \ldots, \Pi_q) $ is said to be a derivation of $ \Gamma \vdash s : \alpha $ if there exist $ b_1, \ldots, $ $ b_q \in \mathcal{M}_{\textrm{fin}}(D) $, $ \Gamma_0, \ldots, $ $ \Gamma_q \in \varPhi $ such that  
\begin{itemize}
\item $ \Pi_0 $ is a derivation of $ \Gamma_0 \vdash c_0 : b_1 \ldots b_q \alpha $;
\item for any $ k \in \{ 1, \ldots, q \} $, $ \Pi_k $ is a derivation of $ \Gamma_k \vdash c_k : b_k $;
\item and $ \Gamma = \sum_{k=0}^q \Gamma_k $.
\end{itemize}
In this case, we set $ \vert (\Pi_0, \ldots, \Pi_q) \vert = \sum_{k=0}^q \vert \Pi_k \vert $.
\end{definition}

As for derivations of closures, we introduce two lemmas about derivations for states, that will be useful for proving 
Propositions \ref{Prop KAM = D} and \ref{prop:normal_principal}.

\begin{lemme}\label{lemme : cas variable dans environnement}
Let $ m, j_0 \in \mathbb{N} $ such that $ 1 \leq j_0 \leq m $. Let $ s = (c'_{j_0}, c_1, \ldots, c_q) \in \mathbb{S} $, $ x_1, \ldots, x_m \in \mathcal{V} $, $ c'_1, \ldots, c'_m \in \mathcal{C} $. For any $ (\Gamma, \alpha) \in \Phi \times D $, if $ \Pi'$ is a derivation of $ \Gamma \vdash s : \alpha $, then there exists a derivation $ \Pi $ of $ \Gamma \vdash ((x_{j_0}, \bigcup_{j=1}^m \{ (x_j, c'j) \} ), c_1, \ldots, c_q) : \alpha $ such that $ \vert \Pi \vert = \vert \Pi' \vert + 1 $.
\end{lemme}

\begin{proof}
We set $ \Pi'= (\Pi'', \Pi_1, \ldots, \Pi_q) $ with $ \Pi'' $ a derivation of $ \Gamma'' \vdash c'_{j_0} : b_1 \ldots b_q \alpha $. We denote by $ \Pi_0 $ the derivation of $ x : [b_1 \ldots b_q \alpha ] \vdash_R x : b_1 \ldots b_q \alpha $. For any $ j \in \{ 1, \ldots, m \} $, we set  
$$ \Pi''_j = \left\{ \begin{array}{ll} (\Pi'') & \textrm{ if $ j = j_0$~;} \\ \epsilon & \textrm{ else.} \end{array} \right.$$
The sequence $ ((\Pi_0, \bigcup_{j=1}^m \{ (x_j, \Pi''_j) \} ), \Pi_1, \ldots, \Pi_q ) $ is a derivation of 
$$ \Gamma \vdash ((x_{j_0}, \bigcup_{j=1}^m \{ (x_j, c'j) \} ), c_1, \ldots, c_q) : \alpha $$ 
and we have 
\begin{eqnarray*}
\vert ((\Pi_0, \bigcup_{j=1}^m \{ (x_j, \Pi''_j) \}) , \Pi_1, \ldots, \Pi_q) \vert & = & \vert \Pi_0 \vert + \sum_{j=1}^m \vert \Pi''_j \vert + \sum_{k=1}^q \vert \Pi_k \vert \allowdisplaybreaks \\
& = & 1 + \vert \Pi'' \vert + \sum_{k=1}^q \vert \Pi_k \vert \allowdisplaybreaks \\
& = & 1 + \vert \Pi' \vert .
\end{eqnarray*}
\end{proof}

\begin{lemme}\label{lemme : derivation etat : cas abstraction et pile non vide}
For any state $ s = ((u, \{ (x, c) \} \cup e), c_1, \ldots, c_q) $, for any derivation $ \Pi' $ of $ \Gamma \vdash s : \alpha $, there exists a derivation $ \Pi $ of $ \Gamma \vdash ((\lambda x. u, e), c, c_1, \ldots, c_q) : \alpha$ such that $ \vert \Pi \vert = \vert \Pi' \vert + 1 $.
\end{lemme}

\begin{proof}
The environment $e$ is of the shape $\bigcup_{j=1}^m \{(x_j, c'_j) \} $ and $\Pi'$ is of the shape $((\Pi'_0, \{ (x, \Pi'') \} \cup \bigcup_{j=1}^m \{ (x_j, \Pi''_j) \} ), \Pi'_1, \ldots, \Pi'_q)$.
We know that $ \Pi'_0 $ is a derivation of $ \Gamma, x_1 : a_1, \ldots, x_m : a_m, x : a \vdash_R u : b_1 \ldots b_q \alpha $, hence there exists a derivation $ \Pi_0 $ of $ x_1 : a_1, \ldots, x_m : a_m \vdash_R \lambda x. u : a b_1 \ldots b_q \alpha $ such that $ \vert \Pi_0 \vert = \vert \Pi'_0 \vert + 1 $. Set $ \Pi = ((\Pi_0, \{ (x_1, \Pi'_1), \ldots, (x_m, \Pi'_m) \} ), \Pi'', \Pi'_1, \ldots, \Pi'_q) $: it is a derivation of $ \Gamma \vdash ((\lambda x. u, e), c, c_1, \ldots, c_q) : \alpha $ and we have 
\begin{eqnarray*}
\vert \Pi \vert & = & \vert \Pi_0 \vert + \sum_{j=1}^m \vert \Pi'_j \vert + \vert \Pi'' \vert + \sum_{k=1}^q \vert \Pi'_k \vert \allowdisplaybreaks \\
& = & \vert \Pi'_0 \vert + 1 + \sum_{j=1}^m \vert \Pi'_j \vert + \vert \Pi'' \vert + \sum_{k=1}^q \vert \Pi'_k \vert \allowdisplaybreaks \\
& = & \vert \Pi' \vert + 1 .
\end{eqnarray*}
\end{proof}

\subsection{Relating size of derivations and execution time}\label{subsection : size of derivations and execution time}

The aim of this subsection is to prove Theorem \ref{Theorem head}, that gives the exact number of steps leading to the principal head normal form by means of derivations in System $R$.

\begin{lemme}\label{KAM <= D cas application}
Let $ ((v)u, e) $ be a closure and let $
(\Gamma, \alpha) \in \Phi \times D $. For any derivation $ \Pi $ of $ \Gamma \vdash
((v)u, e) : \alpha $, there exist  
$ b \in \mathcal{M}_{\textrm{fin}}(D) , $ $ \Gamma', \Gamma'' \in \Phi $, a derivation $
\Pi' $ of $ \Gamma' \vdash (v,e) : (b,
\alpha) $ and a derivation $ \Pi'' $ of
$ \Gamma'' \vdash (u, e) : b $ such that $ \Gamma = \Gamma' + \Gamma'' $ and $ \vert \Pi \vert = \vert \Pi' \vert + \vert \Pi'' \vert + 1 $.
\end{lemme}

\begin{proof}
Set $ e = \{ (x_1, c_1), \ldots, (x_m, c_m) \} $ and  $ \Pi = (\Pi_0, \bigcup_{j=1}^m \{ (x_j, \Pi_j) \}) $ where
\begin{enumerate}
\item $ \Pi_0 $ is a derivation of $ \Gamma_0, x_1 : a_1, \ldots, x_m : a_m \vdash_R (v)u : \alpha $,
\item for $ 1 \leq j \leq m $, $ \Pi_j $ is a derivation of $ \Gamma_j \vdash c_j : a_j $,
\item $ \Gamma = \sum_{j=0}^m \Gamma_j $.
\end{enumerate}
By (i), there exist $b = [\beta_1, \ldots, \beta_p] \in \mathcal{M}_\textrm{fin}(D) $, 
$\Pi_0^0 \in \Delta(v, (\Gamma_0^0, x_1 : a_1', \ldots, x_m : a_m', (b, \alpha)))$ 
and, for $ 1 \leq k \leq p $, a derivation $ \Pi_0^k $ of $ \Gamma_0^k, x_1 : {a_1''}^k, \ldots, x_m : {a_m''}^k \vdash_R u : \beta_k $ such that 
\begin{itemize}
\item $ \Gamma_0 = \sum_{k=0}^p \Gamma_0^k $, 
\item for $ 1 \leq j \leq m $, $ a_j = a_j' + \sum_{k=1}^p {a_j''}^k $ 
\item and $ \vert \Pi_0 \vert = \sum_{k=0}^p \vert \Pi_0^k \vert +1 $. 
\end{itemize}
For any $ j \in \{ 1, \ldots, m \} $, we set $ a_j'' = \sum_{k=1}^p {a_j''}^k $. By (ii), for any $ j \in \{ 1, \ldots, m \} $, there exist $ \Gamma_j', $ $ \Gamma_j''^1, \ldots, \Gamma_j''^m \in \Phi$, a derivation $ \Pi_j' $ of $ \Gamma_j' \vdash c_j : a_j'$ and, for any $k \in \{ 1, \ldots, m \}$, a derivation $ \Pi_j''^k$ of $ \Gamma_j''^k \vdash c_j : a_j''^k $ such that 
\begin{itemize}
\item $ \Gamma_j = \Gamma_j'+ \sum_{k=1}^p \Gamma_j''^p$ 
\item and $ \vert \Pi_j \vert = \vert \Pi_j' \vert + \sum_{k=1}^p \vert \Pi_j''^p \vert$.
\end{itemize}
Set 
\begin{itemize}
\item $ \Gamma'= \Gamma_0^0 + \sum_{j=1}^m \Gamma_j' ;$ 
\item $ \Gamma'' = \sum_{k=1}^p (\Gamma_0^k + \sum_{j=1}^m \Gamma_j''^k) ;$
\item $ \Pi'= (\Pi_0^0, \bigcup_{j=1}^m \{ (x_j, \Pi_j') \})$ 
\item and $ \Pi'' = ( (\Pi_0^1, \bigcup_{j=1}^m \{(x_j, \Pi_j''^1) \} ), \ldots, (\Pi_0^p, \bigcup_{j=1}^m \{ (x_j, \Pi_j''^p) \} ))$.
\end{itemize}
We have
\begin{eqnarray*}
\Gamma & = & \sum_{j=0}^m \Gamma_m \allowdisplaybreaks \\
& & \textrm{   (by (iii))} \allowdisplaybreaks \\
& = & \sum_{k=0}^p \Gamma_0^k + \sum_{j=1}^m (\Gamma_j' + \sum_{k=1}^p \Gamma_j''^k) \allowdisplaybreaks \\
& = & \Gamma' + \Gamma''
\end{eqnarray*}
and
\begin{eqnarray*}
\vert \Pi \vert & = & \sum_{j=0}^m \vert \Pi_j \vert \allowdisplaybreaks \\
& = & \sum_{k=0}^p \vert \Pi_0^k \vert +1 + \sum_{j=1}^m (\vert \Pi_j' \vert + \sum_{k=1}^p \vert \Pi_j''^k \vert) \allowdisplaybreaks \\
& = & \vert \Pi' \vert + \vert \Pi'' \vert + 1 .
\end{eqnarray*} 
\end{proof}

\begin{prop}\label{Prop KAM<=D}
Let $t$ be a head normalizable $\lambda$-term. For any $(\Gamma, \alpha) \in \Phi \times D$, for any $\Pi \in \Delta(t, (\Gamma, \alpha))$, we have $l_h((t, \emptyset) . \epsilon) \leq \vert \Pi \vert$.
\end{prop}

\begin{proof}
By Theorem \ref{theorem : head-normalizable => l fini}, we can prove, by induction on $ l_h(s) $, that for any $s \in \mathbb{S}$ such that $\overline{s}$ is head normalizable, for any $(\Gamma, \alpha) \in \Phi \times D$, for any derivation $\Pi$ of $\Gamma \vdash s:\alpha$, we have $l_h(s) \leq \vert \Pi \vert$. 

The base case is trivial, because we never have $l_h(s)=0 $. The inductive step is divided into five cases:
\begin{itemize}
\item In the case where $ s = (x, e) . \pi $, $ x \in \mathcal{V} $ and $ x \notin \textrm{dom}(e) $, $ l_h(s) = 1 \leq \vert \Pi \vert $.
\item In the case where $ s = ((x_{j_0,}, \bigcup_{j=1}^m \{(x_j, c'_j) \}), c_1, \ldots, c_q) $ and $ 1 \leq j_0 \leq m $, we have $ \Pi =(\Pi_0, \ldots, \Pi_q) $, where $ \Pi_0 = (\Pi_0', \bigcup_{j=1}^m \{ (x_j, \Pi'_j) \} ) $ with
\begin{itemize}
\item $ \Pi_0' $ is a derivation of $ \Gamma'_0, x_1 : a_1, \ldots, x_m : a_m \vdash_R x_{j_0} : b_1 \ldots b_q \alpha $,
\item for any $ j \in \{1, \ldots, m \} $, $ \Pi'_j $ is a derivation of $ \Gamma'_j \vdash c'_j :a_j $,
\item $ \Gamma_0 = \sum_{j=1}^m \Gamma'_j $,
\item for $ 1 \leq k \leq q $, $\Pi_k $ is a derivation of $ \Gamma_k \vdash c_k : b_k $
\item and $ \Gamma = \sum_{k=0}^q \Gamma_k $.
\end{itemize}
Hence $ a'_{j_0} = [b_1 \ldots b_q \alpha ] $. 
The sequence $ (\Pi'_{j_0}, \Pi_1, \ldots, \Pi_q) $ is a derivation of 
$$ \Gamma'_{j_0} + \sum_{k=1}^q \Gamma_k \vdash (c'_{j_0}, c_1, \ldots, c_q) : \alpha . $$
We have
\begin{eqnarray*}
l_h(s) & = & l_h(c_{j_0}', c_1, \ldots, c_q) + 1 \allowdisplaybreaks \\
& \leq & \vert (\Pi'_{j_0}, \Pi_1, \ldots, \Pi_q) \vert + 1 \\
& & \textrm{   (by induction hypothesis)} \allowdisplaybreaks \\
& = & \vert \Pi_{j_0}' \vert + \sum_{k=1}^q  \vert \Pi_k \vert + 1 \allowdisplaybreaks \\
& \leq & \vert \Pi_0 \vert + \sum_{k=1}^q \vert \Pi_k \vert \allowdisplaybreaks \\
& = & \vert \Pi \vert .
\end{eqnarray*}
\item In the case where $ s = ((\lambda x. u, \{ (x_1, c'_1), \ldots, (x_m, c'_m) \}), c', c_1, \ldots, c_q) $, we have $ \Pi = ((\Pi_0', \Pi_0''), \Pi', \Pi_1, \ldots, \Pi_q) $ with
\begin{itemize}
\item $ \Pi_0' $ is a derivation of $ \Gamma_0', x_1 : a_1, \ldots, x_m : a_m \vdash_R \lambda x. u : b' b_1 \ldots b_q \alpha $;
\item $\Pi_0'' = \bigcup_{j=1}^m \{ (x_j, \Pi'_j) \}$ where, for $1 \leq j \leq m$, $ \Pi'_j $ is a derivation of $ \Gamma'_j \vdash c'_j : a_j $;
\item $ \Gamma_0 = \sum_{j=0}^m \Gamma'_j $;
\item $ \Pi' $ is a derivation of $ \Gamma' \vdash b' : c' $;
\item for $ 1 \leq k \leq q $, $ \Pi_k $ is a derivation of $ \Gamma_k \vdash b_k : c_k $.
\end{itemize}
Hence there exists $\Pi'' \in \Delta(u, (\Gamma'_0, x_1:a_1, \ldots, x_m : a_m, x:b', b_1 \ldots b_q \alpha))$
with $ \vert \Pi'_0 \vert = \vert \Pi'' \vert + 1 $. The pair $ (\Pi'', \bigcup_{j=1}^m \{ (x_j, \Pi'_j) \} \cup \{ \Pi' \}) $ is a derivation of 
$$ \Gamma_0 + \Gamma'\vdash (u, \{ (x_1, c'_1), \ldots, (x_m, c'_m), (x, c) \}) : b_1 \ldots b_q \alpha . $$
Hence $ ((\Pi'', \bigcup_{j=1}^m \{ (x_j, \Pi'_j) \} \cup \{ (x, \Pi') \}), \Pi_1, \ldots, \Pi_q) $ is a derivation of 
$$ \Gamma \vdash ((u, \{ (x_1, c'_1), \ldots, (x_m, c'_m), (x, c) \}), c_1, \ldots, c_q) : \alpha . $$
We have
\begin{eqnarray*}
l_h(s) & = & l_h((u, \{ (x_1, c'_1), \ldots, (x_m, c'_m),  (x, c) \}), c_1, \ldots, c_q)+1 \allowdisplaybreaks \\
& \leq & \vert ((\Pi'', \{ \Pi'_1, \ldots, \Pi'_m, \Pi' \}), \Pi_1, \ldots, \Pi_q) \vert + 1 \\
& & \textrm{   (by induction hypothesis)} \allowdisplaybreaks \\
& = & \vert \Pi'' \vert + \sum_{j=1}^m \vert \Pi'_j \vert + \vert \Pi'\vert + \sum_{k=1}^q \vert \Pi_k \vert + 1 \allowdisplaybreaks \\
& = & \vert \Pi'_0 \vert + \vert \Pi_0'' \vert + \vert \Pi'\vert + \sum_{k=1}^q \vert \Pi_k \vert \allowdisplaybreaks \\
& = & \vert \Pi \vert .
\end{eqnarray*}
\item In the case where $ s = (((v)u, e), c_1, \ldots, c_q) $, we have $ \Pi = (\Pi_0, \ldots, \Pi_q) $ with
\begin{itemize}
\item $ \Pi_0 $ is a derivation of $ \Gamma_0 \vdash ((v)u, e) : b_1 \ldots b_q \alpha $;
\item for $ 1 \leq k \leq q $, $ \Pi_k $ is a derivation of $ \Gamma_k \vdash c_k : b_k $;
\item $ \Gamma = \sum_{k=0}^q \Gamma_k $.
\end{itemize}
By Lemma \ref{KAM <= D cas application}, there exist $ b \in \mathcal{M}_{\textrm{fin}}(D)$, $ \Gamma_0', \Gamma_0'' \in \Phi $, a derivation $ \Pi_0' $ of $ \Gamma_0' \vdash (v, e) : b b_1 \ldots b_q \alpha $ and a derivation $ \Pi_0'' $ of $ \Gamma_0'' \vdash (u, e) : b $ such that $ \Gamma_0 = \Gamma_0' + \Gamma_0''$ and $ \vert \Pi_0 \vert = \vert \Pi_0'\vert + \vert \Pi_0'' \vert + 1 $. The sequence $ (\Pi_0', \Pi_0'', \Pi_1, \ldots, \Pi_q) $ is a derivation of $ \Gamma \vdash ((v, e), (u, e), c_1, \ldots, c_q) : \alpha $. We have
\begin{eqnarray*}
l_h(s) & = & l_h((v, e), (u, e), c_1, \ldots, c_q) + 1 \allowdisplaybreaks \\
& \leq & \vert (\Pi_0', \Pi_0'', \Pi_1, \ldots, \Pi_q) \vert + 1 \\
& & \textrm{   (by induction hypothesis)} \allowdisplaybreaks \\
& = & \vert \Pi_0' \vert + \vert \Pi_0'' \vert + \sum_{k=1}^q \vert \Pi_k \vert + 1 \allowdisplaybreaks \\
& = & \vert \Pi_0 \vert + \sum_{k=1}^q \vert \Pi_k \vert \allowdisplaybreaks \\
& = & \vert (\Pi_0, \ldots, \Pi_q) \vert \allowdisplaybreaks \\
& = & \vert \Pi \vert .
\end{eqnarray*}
\item In the case where $ s = (\lambda x. u, \bigcup_{j=1}^m \{ (x_j, c'_j) \}) .\epsilon $, we have $ \Pi = (\Pi_0', \Pi_0'')$ with
\begin{itemize}
\item $ \Pi_0'$ is a derivation of $ \Gamma'_0, x_1 : a_1, \ldots, x_m : a_m \vdash_R \lambda x. u : \alpha $;
\item $\Pi_0'' = \bigcup_{j=1}^m \{ (x_j, \Pi'_j) \}$ where, for $ 1 \leq j \leq m $, $ \Pi'_j $ is a derivation of $ \Gamma'_j \vdash c'_j : a_j $;
\item $ \Gamma = \sum_{j=0}^m \Gamma'_j $.
\end{itemize}
Hence there exists a derivation $ \Pi'' $ of $ \Gamma'_0, x_1 : a_1, \ldots, x_m : a_m, x : b \vdash_R u : \beta $ such that $ \alpha = (b, \beta) $ and $ \vert \Pi_0' \vert = \vert \Pi'' \vert + 1 $. The pair $(\Pi'', \Pi_0'')$ is a derivation of 
$$\Gamma, x : b \vdash (u, \bigcup_{j=1}^m \{ (x_j, c_j) \}) . \epsilon : \beta .$$
We have
\begin{eqnarray*}
l_h(s) & = & l_h((u, \bigcup_{j=1}^m \{ (x_j, c_j) \}) . \epsilon)+1 \allowdisplaybreaks \\
& \leq & \vert (\Pi'', \Pi_0'') \vert + 1 \allowdisplaybreaks \\
& = & \vert \Pi'' \vert + \sum_{j=1}^m \vert \Pi'_j \vert + 1 \allowdisplaybreaks \\
& = & \vert \Pi_0' \vert + \sum_{j=1}^m \vert \Pi'_j \vert  \allowdisplaybreaks \\
& = & \vert \Pi \vert .
\end{eqnarray*} 
\end{itemize}
\end{proof}

\begin{prop}\label{Prop KAM = D}
Let $t$ be a head normalizable $\lambda$-term. There exist $(\Gamma, \alpha) \in \Phi \times D$ and $\Pi \in \Delta(t, (\Gamma, \alpha))$ such that $l_h((t, \emptyset) . \epsilon) = \vert \Pi \vert$.
\end{prop}

\begin{proof}
By Theorem \ref{theorem : head-normalizable => l fini}, we can prove, by induction on $ l_h(s) $, that for any $ s \in \mathbb{S} $ such that $ 
\overline{s} $ is head normalizable, there exist $
(\Gamma, \alpha) $ and a derivation $ \Pi $ of $ \Gamma \vdash 
s : \alpha $ such that we have $ l_h(s) = \vert \Pi \vert $.

The base case is trivial, because we never have $ l_h(s) = 0 $. The inductive step is divided into five cases:
\begin{itemize}
\item In the case where $ s = ((x, e), c_1, \ldots, c_q)$, $ x \in \mathcal{V} $ and $ x \notin \textsf{dom}(e) $, we have $ l_h(s) = 1 $ and there exists a derivation $ \Pi = (\Pi_0, \ldots, \Pi_q) $ of $\Gamma \vdash s : \alpha$, where $\Pi_0$ is a derivation of $ x : [\underbrace{[] \ldots []}_{q \textrm{ times}}\alpha] \vdash (x, e) : \underbrace{[] \ldots []}_{q \textrm{ times}} \alpha $ with $ \vert \Pi_0 \vert = 1 $ and $ \vert \Pi_1 \vert = \ldots = \vert \Pi_q \vert = 0 $. 
\item In the case where $s$ is of the shape $(x, e) . \pi$ with $x \in \textsf{dom}(e)$, apply the induction hypothesis and Lemma \ref{lemme : cas variable dans environnement}.
\item In the case where $s$ is of the shape $((v)u, e) . \pi$,  apply the induction hypothesis and Lemma \ref{KAM = D cas application}.
\item In the case where $s$ is of the shape $(\lambda x. u), e) . \epsilon$, apply the induction hypothesis Lemma \ref{lemme : derivation cloture : cas abstraction}.
\item In the case where $s$ is of the shape $((\lambda x. u), e) . \pi$ with $\pi \not= \epsilon$, apply the induction hypothesis and Lemma \ref{lemme : derivation etat : cas abstraction et pile non vide}.
\end{itemize}
\end{proof}

\begin{definition}
For every $\mathfrak{D} \in \mathcal{P}(\Delta) \cup \mathcal{P}(\Delta^{< \omega})$, we set $\vert \mathfrak{D} \vert = \{ \vert \Pi \vert \: / \: \Pi \in \mathfrak{D} \}$.
\end{definition}

\begin{theorem}\label{Theorem head}
For any $\lambda$-term  $ t $, we have $l_h((t, \emptyset) . \epsilon) = \inf \vert \Delta(t) \vert$.
\end{theorem}

\begin{proof}
We distinguish between two cases.
\begin{itemize}
\item The $\lambda$-term $t$ is not head normalizable: by Theorem \ref{th normalisable de tete} (i), $\inf \vert \Delta(t) \vert = \infty$ and, by Theorem \ref{th Krivine fini => head}, $l_h((t, \emptyset) . \epsilon) = \infty$.
\item The $\lambda$-term $ t $ is head normalizable: apply Proposition \ref{Prop KAM<=D} and Proposition \ref{Prop KAM = D}.
\end{itemize}
\end{proof}

\subsection{Principal typings and $1$-typings}

In the preceding subsection, we related $l_h(t)$ and the size of the derivations of $t$ for any $\lambda$-term $t$. Now, we want to relate $ l_\beta(t) $ and the size of the derivations of $ t $. 
We will show that if the value of $l_\beta(t)$ is finite (i.e. $t$ is normalizable), then $l_\beta(t)$ is the size of the least derivations of $t$ with typings that satisfy a particular property and that, otherwise, there is no such derivation. 
In particular, when $t$ is normalizable, $l_\beta(t)$ is the size of the derivations of $t$ with \emph{$1$-typings}. This notion of \emph{$1$-typing}, defined in Definition~\ref{definition:1-typing}, is a generalization of the notion of \emph{principal typing}.

We recall that a typing $ (\Gamma, \alpha) $ for a $\lambda$-term is a principal typing if all other typings for the same $\lambda$-term can be derived from $ (\Gamma, \alpha) $ by some set of operations. 
The work of \cite{CDV} could be adapted in order to show that all normal $\lambda$-terms have a principal typing in System $ R $ if $ A $ is infinite: the operations are substitution (see 
Definition~\ref{definition : substitution}) and expansion (complicated to define); the only difference with \cite{CDV} is that we should have to consider $0$-expansions too and not only $n$-expansions for $n \geq 1$. 

\begin{definition}
The typing rules for deriving \emph{principal typings} of normal $\lambda$-terms are the following:
\begin{center}
\AxiomC{}
\RightLabel{$ \gamma \in A $}
\UnaryInfC{$ x : [ \gamma ] \vdash_P x :
\gamma $}
\DisplayProof
\end{center}
\begin{center}
\AxiomC{$ \Gamma, x : a \vdash_P t : \alpha $}
\UnaryInfC{$ \Gamma \vdash_P \lambda x. t : (a, \alpha) $}
\DisplayProof
\end{center}
\begin{center}
\AxiomC{$ \Gamma_1 \vdash_P u_1 : \alpha_1 $}
\AxiomC{$ \ldots $}
\AxiomC{$ \Gamma_n \vdash_P u_n : \alpha_n $}
\RightLabel{$ (\ast) $}
\TrinaryInfC{$ \sum_{i=1}^n \Gamma_i + \{ (x, [ [
\alpha_1] \ldots
[\alpha_n ]
\gamma ] ) \} \vdash_P (x)u_1 \ldots u_n : \gamma $}
\DisplayProof
\end{center}
\begin{center}
$ (\ast) $ the atoms in $\Gamma_j$ are disjoint from those in $\Gamma_k$ if $j \not= k$ and $ \gamma \in
A $ does not appear in the $ \Gamma_i $
\end{center}
A \emph{principal typing} of a normalizable $\lambda$-term is a principal typing of its normal form.
\end{definition}

The reader acquainted with the concept of \emph{experiment} on proof nets in linear logic could notice that a principal typing of a normal $\lambda$-term is the same thing as the result of what \cite{Lorenzo} calls \emph{an injective $ 1 $-experiment} of the proof net obtained by the translation of this $\lambda$-term mentioned in Subsection \ref{subsection : relating types and semantics}.

The notion of $ 1 $-typing is more general than the notion of principal typing: it is the result of a \emph{$ 1 $-experiment} (not necessarily injective).

\begin{definition}\label{definition:1-typing}
The typing rules for deriving \emph{$1$-typings} of normal $\lambda$-terms are the following:
\begin{center}
\AxiomC{}
\RightLabel{$ \gamma \in A $}
\UnaryInfC{$ x : [ \gamma ] \vdash_1 x :
\gamma $}
\DisplayProof
\end{center}
\begin{center}
\AxiomC{$ \Gamma, x : a \vdash_1 t : \alpha $}
\UnaryInfC{$ \Gamma \vdash_1 \lambda x. t : (a, \alpha) $}
\DisplayProof
\end{center}
\begin{center}
\AxiomC{$ \Gamma_1 \vdash_1 u_1 : \alpha_1 $}
\AxiomC{$ \ldots $}
\AxiomC{$ \Gamma_n \vdash_1 u_n : \alpha_n $}
\TrinaryInfC{$ \sum_{i=1}^n \Gamma_i + \{ (x, [ [
\alpha_1] \ldots
[\alpha_n ]
\gamma ] ) \} \vdash_1 (x)u_1 \ldots u_n : \gamma $}
\DisplayProof
\end{center}
A \emph{$1$-typing} of a normalizable $\lambda$-term is a $1$-typing of its normal form.
\end{definition}

Note that if $t$ is a normalizable $\lambda$-term and $(\Gamma, \alpha)$ is a $1$-typing of $t$, then $(\Gamma, \alpha) \in \Phi^\textsf{ex} \times D^\textsf{ex}$; more precisely, a typing $(\Gamma, \alpha)$ of a normalizable $\lambda$-term is a $1$-typing if, and only if, every multiset in negative occurence in $\Gamma$ (resp. in positive occurrence in $\alpha$) is a singleton.

\begin{lemme}\label{lemme : normal : bound}
Let $ (x, e) $ be a closure and let $ \Gamma \in \Phi $ such that $\Gamma \in \Phi^\textsf{ex}$. Assume that there exists a derivation of $ \Gamma \vdash (x, e) : b_1 \ldots b_q \alpha $, with $ x \notin \textrm{dom}(e) $, then for any $ k \in \{ 1, \ldots, q \} $, we have $ b_k \not= [] $.
\end{lemme}

\begin{proof}
Let $ \Pi $ be such a derivation. Set $ e = \bigcup_{j=1}^m \{ (x_j, c_j) \} $. We have $ \Pi = (\Pi_0, $ $ \bigcup_{j=1}^m \{ (x_j, \Pi_j) \}) $, where
\begin{enumerate}
\item $ \Pi_0 $ is a derivation of $ \Gamma_0, x_1 : a_1, \ldots, x_m : a_m \vdash_R x : b_1 \ldots b_q \alpha $;
\item for $ j \in \{ 1, \ldots, m \} $, $ \Pi_j $ is a derivation of $ \Gamma_j \vdash c_j : a_j $;
\item and $ \Gamma = \sum_{j=0}^m \Gamma_j $.
\end{enumerate}
By (i), since $ x \notin \textrm{dom}(e) $, $ \Gamma_0(x) = [b_1 \ldots b_q \alpha] $. Hence, by (iii), if there existed $ k \in \{ 1, \ldots, q \} $ such that $ b_k =[] $, then we should have $\Gamma \notin \Phi^\textsf{ex}$.
\end{proof}

\begin{prop}\label{prop:normal_bound}
Let $ t $ be a normalizable $\lambda$-term. If $ \Pi $ is a derivation of $ \Gamma \vdash_R t : \alpha $ and $ (\Gamma, \alpha) \in \Phi^\textsf{ex} \times D^\textsf{ex}$, then we have $ l_\beta((t, \emptyset) . \epsilon) \leq \vert \Pi \vert $.
\end{prop}

\begin{proof}
By Theorem \ref{theorem : normalizable => l' finite}, we can prove, by induction on $ l_\beta(s) $, that for any $ s = (c_0, \ldots, c_q) \in \mathbb{S} $ such that $ (\overline{c_0})\overline{c_1} \ldots \overline{c_q} $ is normalizable, for any $(\Gamma, \alpha) \in \Phi \times D$,  if $ \Pi $ is a derivation of $ \Gamma \vdash s : \alpha $ and $ (\Gamma, \alpha) \in \Phi^\textsf{ex} \times D^\textsf{ex}$, then we have $ l_\beta(s) \leq \vert \Pi \vert $.

In the case where $ s = ((x, e), c_1, \ldots, c_q) $ and $ x \notin \textrm{dom}(e) $, we apply Lemma \ref{lemme : normal : bound}.
\end{proof}

\begin{prop}\label{prop:normal_principal}
Assume that $ t $ is a normalizable $\lambda$-term and that $ (\Gamma, \alpha) $ is a $1$-typing of $t$. Then there exists a derivation $ \Pi $ of $ \Gamma \vdash_R t : \alpha $ such that $ l_\beta((t, \emptyset) . \epsilon) = \vert \Pi \vert $.
\end{prop}

\begin{proof}
By Theorem \ref{theorem : normalizable => l' finite}, we can prove, by induction on $ l_\beta(s) $, that for any $s \in \mathbb{S}$ such that $\overline{s}$ is normalizable and for any $1$-typing $(\Gamma, \alpha)$ of $\overline{s}$, there exists a derivation $\Pi$ of $\Gamma \vdash s:\alpha$ such that $l_\beta(s) = \vert \Pi \vert$. 

The base case is trivial, because we never have $ l_\beta(s) = 0 $. The inductive step is divided into five cases:
\begin{itemize}
\item In the case where $ s = ((x, e), c_1, \ldots, c_q) $ and $ x \notin \textrm{dom}(e) $, $ (\Gamma, \alpha) $ is a $ 1 $-typing of $ (x) t_1 \ldots t_q $, where $ t_1, \ldots, t_q $ are the respective normal forms of $ \overline{c_1}, $ $ \ldots, $ $ \overline{c_q} $, hence there exist $ \Gamma_1, \ldots, \Gamma_q $, $ \alpha_1, \ldots, \alpha_q $ such that 
\begin{itemize}
\item $ \Gamma = \sum_{k=1}^q \Gamma_k + \{ (x,[[\alpha_1] \ldots [\alpha_q] \alpha ]  )\} $
\item and $ (\Gamma_1, \alpha_1), \ldots, (\Gamma_q, \alpha_q) $ are $ 1 $-typings of $ t_1, \ldots, t_q $ respectively.
\end{itemize}
By induction hypothesis, there exist $ q $ derivations $ \Pi_1, \ldots, \Pi_q $ of $ \Gamma_1 \vdash_R t_1 : \alpha_1 ,$ $ \ldots, $ $ \Gamma_q \vdash_R t_q : \alpha_q $ respectively. We denote by $ x_1, $ $ \ldots , $ $ x_m $ the elements of $ \textrm{dom}(e) $. We denote by $ \Pi_0 $ the derivation of 
$$ x: [[\alpha_1] \ldots [\alpha_q] \alpha] \vdash_R x : \alpha . $$ 
Set $ \Pi = ((\Pi_0, \bigcup_{j=1}^m \{ (x_j, \epsilon) \}), \Pi_1, \ldots, \Pi_q) $: it is a derivation of $ \Gamma \vdash_R \overline{s} : \alpha $ and we have 
\begin{eqnarray*}
l_\beta(s) & = & \sum_{k=1}^q l_\beta(c_k) + 1 \allowdisplaybreaks \\
& = & \sum_{k=1}^q \vert \Pi_k \vert + 1 \\
& & \textrm{   (by induction hypothesis}) \allowdisplaybreaks \\
& = & \vert \Pi_0 \vert + \sum_{k=1}^q \vert \Pi_k \vert \allowdisplaybreaks \\
& = & \vert \Pi \vert .
\end{eqnarray*}
\item In the case where $s$ is of the shape $(x, e) . \pi$ with $x \in \textsf{dom}(e)$, apply the induction hypothesis and Lemma \ref{lemme : cas variable dans environnement}.
\item In the case where $s$ is of the shape $((v)u, e) . \pi $, apply the induction hypothesis and Lemma \ref{KAM = D cas application}.
\item In the case where $s$ is of the shape $(\lambda x. u, e) . \epsilon $, apply the induction hypothesis and Lemma \ref{lemme : derivation cloture : cas abstraction}.
\item In the case where $s$ is of the shape $(\lambda x. u, e) . \pi $ with $\pi \not= \epsilon $, apply the induction hypothesis and Lemma \ref{lemme : derivation etat : cas abstraction et pile non vide}.
\end{itemize}
\end{proof}

\begin{definition}
For any $\lambda$-term $t$, we set $\Delta^\textsf{ex}(t) = \bigcup_{(\Gamma, \alpha) \in \Phi^\textsf{ex} \times D^\textsf{ex}} \Delta(t, (\Gamma, \alpha))$.
\end{definition}

\begin{theorem}\label{theorem:normal_bound}
For any $\lambda$-term $t$, we have $l_\beta((t, \emptyset) . \epsilon) = \inf \vert \Delta^\textsf{ex}(t) \vert $.
\end{theorem}

\begin{proof}
We distinguish between two cases.
\begin{itemize}
\item The $\lambda$-term $ t $ is not normalizable: by Theorem \ref{th normalisable de tete} (ii), $\inf \vert \Delta^\textsf{ex}(t) \vert = \infty$ and, by Theorem \ref{theorem : normalizable => l' finite}, $l_\beta((t, \emptyset) . \epsilon) = \infty$.
\item The $\lambda$-term $ t $ is normalizable: apply Proposition \ref{prop:normal_bound} and Proposition \ref{prop:normal_principal}.
\end{itemize}
\end{proof}

\subsection{Relating semantics and execution time}\label{semantics}

In this subsection, we prove the first truly semantic measure of execution time of this paper by bounding (by purely semantic means, i.e. without considering derivations) the number of steps of the computation of the principal head normal form (Theorem \ref{theorem : semantics bounds}).

We define the size $ \vert \alpha \vert $ of any $\alpha \in D$ using an auxiliary function $\textsf{aux}$.

\begin{definition}
For any $ \alpha \in D$, we define $ \vert \alpha \vert $ and $ \textsf{aux}(\alpha) $ by induction on $\textsf{depth}(\alpha)$,:
\begin{itemize}
\item if $ \alpha \in A $, then $ \vert \alpha \vert = 1 $ and $ \textsf{aux}(\alpha) = 0 $;
\item if $ \alpha = ([\alpha_1, \ldots, \alpha_n], \alpha_0) $, then 
\begin{itemize}
\item $ \vert \alpha \vert = \sum_{i=1}^n \textsf{aux}(\alpha_i) + \vert \alpha_0 \vert +1$
\item and $\textsf{aux}(\alpha) = \sum_{i=1}^n \vert \alpha_i \vert + textsf{aux}(\alpha_0)+1 .$
\end{itemize} \end{itemize}
For any $a=[\alpha_1, \ldots, \alpha_n] \in \mathcal{M}_\textrm{fin}(D)$, we set $\vert a \vert = \sum_{i=1}^n \vert \alpha_i \vert $ and $ \textsf{aux}(\alpha) = \sum_{i=1}^n \textsf{aux}(\alpha_i) $.
\end{definition}

Notice that for any $\alpha \in D$, the size $\vert \alpha \vert$ of $\alpha$ is the sum of the number of positive occurrences of atoms in $\alpha$ and of the number of commas separating a multiset of types and a type.

\begin{ex}\label{ex:size}
Let $\gamma \in A$. Set $\alpha = ([\gamma], \gamma)$ and $a = [\underbrace{\alpha, \ldots, \alpha}_{n \textrm{ times}}]$. We have $\vert (a, \alpha) \vert = 2n+3$.
\end{ex}

\begin{lemme}\label{lemme : taille d'un type clos}
For any $\lambda$-term $ u $, if there exists a derivation $ \Pi $ of $ x_1 : a_1, \ldots, x_m : a_m \vdash_R u : \alpha $, then $ \vert a_1 \ldots a_m \alpha \vert = \textsf{aux}(a_1 \ldots a_m \alpha) $.
\end{lemme}

\begin{proof}
By induction on $ \Pi $.
\end{proof}

\begin{lemme}\label{lemme : taille de Pi <= taille de alpha}
Let $ v $ be a normal $\lambda$-term and let $ \Pi $ be a derivation of $ x_1 : a_1, \ldots, x_m : a_m \vdash_R v : \alpha $. Then we have $ \vert \Pi \vert \leq \vert a_1 \ldots a_m \alpha \vert $.
\end{lemme}

\begin{proof}
By induction on $ v $.
\end{proof}

\begin{theorem}\label{theorem : semantics bounds}
Let $ v $ and $ u $ be two closed normal $\lambda$-terms. Assume $ (a, \alpha) \in \llbracket v \rrbracket $ and $ \textsf{Supp}(a) \subseteq \llbracket u \rrbracket $. 
\begin{enumerate}
\item\label{theorem:semantics bounds - head} We have 
$ l_h(((v)u, \emptyset) . \epsilon) \leq 2 \vert a \vert + \vert \alpha \vert +2 . $
\item\label{theorem:semantics bounds - normal} If, moreover, $\alpha \in D^\textsf{ex}$, then we have 
$$ l_\beta(((v)u, \emptyset) . \epsilon) \leq 2 \vert a \vert + \vert \alpha \vert +2 . $$
\end{enumerate}
\end{theorem}

\begin{proof}
Set $ a = [\alpha_1, \ldots, \alpha_n] $. There exist a derivation $ \Pi_0 $ of $ \vdash_R v : (a, \alpha) $ and $ n $ derivations $ \Pi_1, \ldots, \Pi_n $ of $ \vdash_R u : \alpha_1 , $ $ \ldots, $ $ \vdash_R u : \alpha_n $ respectively. Hence there exists a derivation $ \Pi $ of $ \vdash_R (v)u : \alpha $ such that $ \vert \Pi \vert = \sum_{i=0}^n \vert \Pi_i \vert + 1 $. 
\begin{enumerate}
\item We have
\begin{eqnarray*}
l_h(((v)u, \emptyset) . \epsilon) & \leq & \sum_{i=0}^n \vert \Pi_i \vert +1 \\
& & \textrm{   (by Proposition \ref{Prop KAM<=D})} \allowdisplaybreaks \\
& \leq & \vert (a, \alpha) \vert + \sum_{i=1}^n \vert \alpha_i \vert + 1 \\
& & \textrm{   (by Lemma \ref{lemme : taille de Pi <= taille de alpha})} \allowdisplaybreaks \\
& = & \sum_{i=1}^n \textsf{aux}(\alpha_i) + \vert \alpha \vert + 1 + \vert a \vert + 1 \allowdisplaybreaks \\
& = & \sum_{i=1}^n \vert \alpha_i \vert + \vert \alpha \vert + 1 + \vert a \vert + 1 \\
& & \textrm{   (by Lemma \ref{lemme : taille d'un type clos})} \allowdisplaybreaks \\
& = & 2 \vert a \vert + \vert \alpha \vert +2 .
\end{eqnarray*}
\item The only difference with the proof of (i) is that we apply Proposition \ref{prop:normal_bound} instead of Proposition \ref{Prop KAM<=D}.
\end{enumerate}
\end{proof}

\subsection{The exact number of steps}\label{subsection:exact_number}

This subsection is devoted to giving the exact number of steps of computation by purely semantic means. For arbitrary points $(a, \alpha) \in \llbracket v \rrbracket$ such that $a \in \mathcal{M}_\textsf{fin}(\llbracket u \rrbracket)$, it is clearly impossible to obtain an equality in Theorem \ref{theorem : semantics bounds}, because there exist such points with different sizes.

The only equalities we have by now are Theorem \ref{Theorem head} and Theorem \ref{theorem:normal_bound}, which use the size of the derivations. A first idea is then to look for points $(a, \alpha) \in \llbracket v \rrbracket$ such that $a \in \mathcal{M}_\textsf{fin}(\llbracket u \rrbracket)$ with $\vert (a, \alpha) \vert$ equals to the sizes of the derivations used in these theorems. But there are cases in which such points do not exist.

A more subtle way out is nevertheless possible, and here is where the notions of equivalence between derivations and of substitution defined in Subsection \ref{subsection : resource terms} come into the picture. More precisely, using the notion of substitution, Proposition~\ref{proposition : cas A infini} (the only place where we use the non-finiteness of the set $A$ of atoms through Fact~\ref{fact:ground deduction} and Lemma~\ref{lemma:size(type) and size(derivation)}) shows how to find, for any $\beta \in \llbracket t \rrbracket$, an element $\alpha \in \llbracket t \rrbracket$ such that $\vert \alpha \vert = \min \vert \Delta(t, \beta) \vert$.

We remind that $ A = D \setminus (\mathcal{M}_{\textrm{fin}}(D) \times D) $. The equivalence relation $ \sim $ has been defined in Definition~\ref{definition : equivalence relation} and the notion of substitution has been defined in Definition~\ref{definition : substitution}. We recall that we denote by $\mathcal{S}$ the set of substitutions.

\begin{fait}\label{fact:ground deduction}
Let $ v $ be a normal $\lambda$-term and let $ \Pi $ be a derivation of 
$$x_1 : b_1, \ldots, x_m : b_m \vdash_R v : \beta .$$ 
There exist $ a_1, \ldots, a_m, \alpha $ and a derivation $ \Pi' $ of $ x_1 : a_1, \ldots, x_m : a_m \vdash_R v : \alpha $ such that $ \Pi' \sim \Pi $ and $ \vert \Pi' \vert + m = \vert a_1 \ldots a_m \alpha \vert $. If, moreover, $ A $ is infinite, then we can choose $ \Pi' $ in such a way that there exists a substitution $ \sigma $ such that $ \overline{\sigma}(a_1) = b_1, \ldots, $ $ \overline{\sigma}(a_m) = b_m $ and $ \sigma(\alpha) = \beta $.
\end{fait}

\begin{proof}
By induction on $ v $.
\end{proof}

In the case where $ A $ is infinite, the derivation $ \Pi' $ of the lemma is what \cite{CDV} calls a \emph{ground deduction for $ v $}.

\begin{definition}
For every $X \in \mathcal{P}(D) \cup \mathcal{P}(\mathcal{M}_\textrm{fin}(D))$, we set $\vert X \vert = \{ \vert \alpha \vert \: / \: \alpha \in X \}$.
\end{definition}

\begin{lemme}\label{lemma:size(type) and size(derivation)}
Assume $A$ is infinite. Let $ t $ be a closed normal $\lambda$-term, let $ \beta \in D $ and let $ \Pi \in \Delta(t, \beta)$. 
Then we have
$$ \vert \Pi \vert = \min \vert \{ \alpha \in D \: / \: (\exists \Pi' \in \Delta(t, \alpha)) (\exists \sigma \in \mathcal{S}) (\Pi' \sim \Pi \textrm{ and } \sigma(\alpha) = \beta) \} \vert . $$
\end{lemme}

\begin{proof}
Apply Lemma \ref{lemme : taille de Pi <= taille de alpha} and Fact \ref{fact:ground deduction}. 
\end{proof}

\begin{prop}\label{proposition : cas A infini}
Assume $ A $ is infinite. Let $ t $ be a closed normal $\lambda$-term and let $ \beta \in \llbracket t \rrbracket $. We have 
$\min \vert \Delta(t, \beta) \vert = \min \vert \{ \alpha \in \llbracket t \rrbracket \: / \: (\exists \sigma \in \mathcal{S}) \sigma(\alpha) = \beta \} \vert .$
\end{prop}

\begin{proof}
Set $m = \min \vert \Delta(t, \beta) \vert $ and $n = \min \vert \{ \alpha \in \llbracket t \rrbracket \: / \: (\exists \sigma \in \mathcal{S}) \sigma(\alpha) = \beta \} \vert  . $ 

First, we prove that $m \leq n$. Let $\alpha \in \llbracket t \rrbracket$ such that we have $(\exists \sigma \in \mathcal{S}) \sigma(\alpha) = \beta$. 
By Theorem~\ref{Theo : semantique = types}, $\Delta(t, \alpha) \not= \emptyset$: let $\Pi' \in \Delta(t, \alpha)$. 
By Proposition \ref{proposition : substitution and equivalence class}, there exists $\Pi \in \Delta(t, \beta)$ such that $\Pi \sim \Pi'$. By Lemma \ref{lemme : taille de Pi <= taille de alpha}, we have $\vert \Pi' \vert \leq \vert \alpha \vert$. Hence we obtain 
$ m \leq \vert \Pi \vert = \vert \Pi' \vert \leq \vert \alpha \vert . $

Now, we prove the inequality $n \leq m$. Let $\Pi \in \Delta(t, \beta)$.
\begin{eqnarray*}
n & = & \min \vert \{ \alpha \in D \: / \: (\exists \Pi' \in \Delta(t, \alpha)) (\exists \sigma \in \mathcal{S}) \sigma(\alpha) = \beta \} \vert \\
& & \textrm{   (by Theorem~\ref{Theo : semantique = types})} \\
& \leq & \min \vert \{ \alpha \in D \: / \: (\exists \Pi' \in \Delta(t, \alpha)) \: (\exists \sigma \in \mathcal{S}) \: (\Pi' \sim \Pi \textrm{ and } \sigma(\alpha) = \beta) \} \vert \\
& = & \vert \Pi \vert \\
& & \textrm{   (by Lemma \ref{lemma:size(type) and size(derivation)}).}
\end{eqnarray*}
\end{proof}

\begin{corollaire}\label{cor : cas A infini}
Assume $ A $ is infinite. Let $ t $ be a closed normal $\lambda$-term and let \mbox{$ b \in \mathcal{M}_\textrm{fin}(\llbracket t \rrbracket) $.} We have 
$\min \vert \Delta(t, b) \vert = \min \vert \{ a \in \mathcal{M}_{\textrm{fin}}(\llbracket t \rrbracket) \: / \: (\exists \sigma \in \mathcal{S}) \:  \overline{\sigma}(a) = b \} \vert .$
\end{corollaire}

The point of Theorem \ref{theorem semantics} is that the number of steps of the computation of the (principal head) normal form of $(v)u$, where $v$ and $u$ are two closed normal $\lambda$-terms, can be determined from $\llbracket v \rrbracket$ and $\llbracket u \rrbracket$.

\begin{definition}
For any $X, Y \subseteq D$, we denote by $\mathcal{U}(X, Y)$ the set
$$\{ ((a, \alpha), a') \in (X \setminus A) \times \mathcal{M}_\textrm{fin}(Y) \: / \: (\exists \sigma \in \mathcal{S}) \: \overline{\sigma}(a) = \overline{\sigma}(a') \} $$
and by $\mathcal{U}^\textsf{ex}(X, Y)$ the set
$$\left\lbrace ((a, \alpha), a') \in (X \setminus A) \times \mathcal{M}_\textrm{fin}(Y) \: / \begin{array}{l} (\exists \sigma \in \mathcal{S}) (\overline{\sigma}(a) = \overline{\sigma}(a') \textrm{ and } \sigma(\alpha) \in D^\textsf{ex}) \end{array} \right\rbrace  .$$
\end{definition}

\begin{theorem}\label{theorem semantics}
Assume $ A $ is infinite. For any two closed normal $\lambda$-terms $u$ and $v$, we have
\begin{enumerate}
\item 
$l_h(((v)u, \emptyset) . \epsilon) = \inf \{ \vert (a, \alpha) \vert + \vert a' \vert + 1 \: / \: ((a, \alpha), a') \in \mathcal{U}(\llbracket v \rrbracket, \llbracket u \rrbracket) \} ;$
\: \: \: 
\item 
$l_\beta(((v)u, \emptyset) . \epsilon) = \inf \{ \vert (a, \alpha) \vert + \vert a' \vert + 1 \: / \: ((a, \alpha), a') \in \mathcal{U}^\textsf{ex}(\llbracket v \rrbracket, \llbracket u \rrbracket) \} .$
\end{enumerate}
\end{theorem}

\begin{proof}
\begin{enumerate}
\item We distinguish between two cases.
\begin{itemize}
\item If $ \Delta((v)u) = \emptyset$, then Theorem \ref{Theorem head} shows that $l_h(((v)u, \emptyset) . \epsilon) = \infty$ and Theorem~\ref{Theo : semantique = types} and Proposition \ref{proposition : substitution and equivalence class} show that $\mathcal{U}(\llbracket v \rrbracket, \llbracket u \rrbracket) = \emptyset$.
\item Else, we have
\begin{eqnarray*}
& & l_h(((v)u, \emptyset) . \epsilon) \\
& = & \min \{ \vert \Pi \vert + \vert \Pi' \vert + 1 \: / \: (\Pi, \Pi') \in \bigcup_{(b, \beta) \in \mathcal{M}_\textrm{fin}(D) \times D} (\Delta(v, (b, \beta)) \times \Delta(u, b)) \} \\
& & \textrm{   (by Theorem \ref{Theorem head})} \allowdisplaybreaks \\
& = & \min \{ \vert (a, \alpha) \vert + \vert a' \vert + 1 \: / \: ((a, \alpha), a') \in \mathcal{U}(\llbracket v \rrbracket, \llbracket u \rrbracket) \} \\
& & \textrm{   (by applying Proposition \ref{proposition : cas A infini} and Corollary \ref{cor : cas A infini}, and by noticing} \\
& & \textrm{    that the atoms in $a$ can be assumed distinct of those in $a'$).}
\end{eqnarray*}
\end{itemize}
\item We distinguish between two cases.
\begin{itemize}
\item If $ \Delta^\textsf{ex}((v)u) = \emptyset$, then Theorem \ref{theorem:normal_bound} shows that 
$l_\beta(((v)u, \emptyset) . \epsilon) = \infty$ and Theorem \ref{Theo : semantique = types} and Proposition \ref{proposition : substitution and equivalence class} show that $\mathcal{U}^\textsf{ex}(\llbracket v \rrbracket, \llbracket u \rrbracket) = \emptyset$.
\item Else, we have
\begin{eqnarray*}
& & l_\beta(((v)u, \emptyset) . \epsilon) \\
& = & \min \{ \vert \Pi \vert + \vert \Pi' \vert + 1 \: / \: (\Pi, \Pi') \in \bigcup_{(b, \beta) \in \mathcal{M}_\textrm{fin}(D) \times D^\textsf{ex}} (\Delta(v, (b, \beta)) \times \Delta(u, b)) \} \\
& & \textrm{   (by Theorem \ref{theorem:normal_bound})} \allowdisplaybreaks \\
& = & \min \{ \vert (a, \alpha) \vert + \vert a' \vert + 1 \: / \: ((a, \alpha), a') \in \mathcal{U}^\textsf{ex}(\llbracket v \rrbracket, \llbracket u \rrbracket) \} \\
& & \textrm{   (by applying Proposition \ref{proposition : cas A infini} and Corollary \ref{cor : cas A infini}, and by noticing} \\
& & \textrm{    that the atoms in $a$ can be assumed distinct of those in $a'$).}
\end{eqnarray*}
\end{itemize}
\end{enumerate}
\end{proof}

\begin{ex}
Set $ v = \lambda x. (x)x $ and $ u = \lambda y. y $. Let $ \gamma_0, \gamma_1 \in A $. Set 
\begin{itemize}
\item $\alpha = \gamma_0$;
\item $a = [\gamma_0, ([\gamma_0], \gamma_0)]$;
\item $a' = [([\gamma_1], \gamma_1), ([\gamma_2], \gamma_2)]$.
\end{itemize}
Let $ \sigma $ be a substitution such that $ \sigma(\gamma_0) = ([\gamma_0], \gamma_0) $, $ \sigma(\gamma_1) = \gamma_0 $ and $ \sigma(\gamma_2) = \alpha $. We have 
\begin{itemize}
\item $ (a, \alpha) \in \llbracket v \rrbracket $;
\item $ \textsf{Supp}(a') \subseteq \llbracket u \rrbracket$;
\item $ \overline{\sigma}(a) = \overline{\sigma}(a')$;
\item $ \vert (a, \alpha) \vert = 4 $ and $ \vert a' \vert = 4 $.
\end{itemize}
By Example \ref{example : (I)I}, we know that we have $ l_h(((v)u, \emptyset) . \epsilon) = 9 $. And we have $ \vert (a, \alpha) \vert + \vert a' \vert + 1 = 9 $.
\end{ex}

The following example shows that the assumption that $A$ is infinite is necessary.

\begin{ex}
Let $n$ be a nonzero integer. Set $I = \lambda y. y$ and $v = \lambda x. (x) \underbrace{I \ldots I}_{n \textrm{ times}}$. We have
$$ \mathcal{U}(\llbracket v \rrbracket, \llbracket I \rrbracket) \subseteq \{ (([[([\alpha_1], \alpha_1)] \ldots [([\alpha_n], \alpha_n)] \alpha], \alpha), [([\alpha_0], \alpha_0)]) \: / \: \alpha_0, \ldots, \alpha_n, \alpha \in D \} . $$
Let $\gamma_0, \ldots, \gamma_n, \delta \in A$ distinct. We have
$$ (([[([\gamma_1], \gamma_1)] \ldots [([\gamma_n], \gamma_n)] \delta], \delta), [([\gamma_0], \gamma_0)]) \in \mathcal{U}^\textsf{ex}(\llbracket v \rrbracket, \llbracket I \rrbracket) . $$
Hence, for any $\alpha_0, \ldots, \alpha_n, \alpha \in D$, if there exists $i \in \{ 0, \ldots, n \}$ such that $\alpha_i \notin A$, then we have
\begin{eqnarray*}
\vert ([[([\alpha_1], \alpha_1)] \ldots [([\alpha_n], \alpha_n)] \alpha], \alpha) \vert + \vert [([\alpha_0], \alpha_0)] \vert + 1 & > & l_\beta(((v)u, \emptyset) . \epsilon) \\
& = & l_h(((v)u, \emptyset) . \epsilon).  
\end{eqnarray*}
On the other hand, if $\gamma_0, \ldots, \gamma_n \in A$, $\alpha \in D$ and there exist $i, j \in \{ 0, \ldots, n \}$ such that $i \not= j$ and $\gamma_i = \gamma_j$, then
$$ (([[([\gamma_1], \gamma_1)] \ldots [([\gamma_n], \gamma_n)] \alpha], \alpha), [([\gamma_0], \gamma_0)]) \notin \mathcal{U}(\llbracket v \rrbracket, \llbracket I \rrbracket) . $$
All this shows that if $\textsf{Card}(A) = n$, then we do not have
$$l_h(((v)I, \emptyset) . \epsilon) = \inf \{ \vert (a, \alpha) \vert + \vert a' \vert + 1 \: / \: ((a, \alpha), a') \in \mathcal{U}(\llbracket v \rrbracket, \llbracket I \rrbracket) \} ;$$
neither
$$l_\beta(((v)I, \emptyset) . \epsilon) = \inf \{ \vert (a, \alpha) \vert + \vert a' \vert + 1 \: / \: ((a, \alpha), a') \in \mathcal{U}^\textsf{ex}(\llbracket v \rrbracket, \llbracket I \rrbracket) \} .$$
\end{ex}

Note that, as the following example illutrates, the non-idempotency is crucial.

\begin{ex}
For any integer $ n \geq 1 $, set $ \overline{n} = \lambda f. \lambda x. \underbrace{(f) \ldots (f)}_{n \textrm{ times}} x $ and $ I = \lambda y. y $. Let $ \gamma \in A $. Set $ \alpha = ([\gamma], \gamma) $ and $a = [\underbrace{\alpha, \ldots, \alpha}_{n \textrm{ times}}] $. We have $ (a, \alpha) \in \llbracket \overline{n} \rrbracket $ and $ \alpha \in \llbracket I \rrbracket $. We have $ l_h(((\overline{n})I, \emptyset) . \epsilon) = 4(n+1) = 2n+3 + 2n + 1 = \vert (a, \alpha) \vert + \vert a \vert + 1 $ (see Example \ref{ex:size}). 
But with \emph{idempotent} types (as in System $\mathcal{D}$), for any integers $p, q \geq 1$, we would have $\mathcal{U}(\overline{p}, I) = \mathcal{U}(\overline{q}, I)$ (any Church integer $\overline{n}$, for $n \geq 1$ has type $((\gamma \rightarrow \gamma) \rightarrow (\gamma \rightarrow \gamma))$ in System $\mathcal{D}$).
\end{ex}